\documentclass[conference]{IEEEtran}
\usepackage{graphicx}
\pdfoutput=1
\usepackage{balance}
\usepackage[T1]{fontenc}
\usepackage{aecompl}
\usepackage[cmex10]{amsmath}
\usepackage{amssymb}
\usepackage{amsfonts}
\usepackage{amsthm}
\usepackage{color}
\usepackage{hyperref}
\usepackage{wrapfig}
\usepackage{euscript}
\usepackage{epsfig}
\usepackage{xspace}
\usepackage{relsize}
\usepackage{colortbl}
\usepackage[update,prepend]{epstopdf}
\usepackage{subfigure,color}
\newtheorem{example}{\textbf{Example}}

\newtheorem{theorem}{\textbf{Theorem}}
\newtheorem{lemma}{\textbf{Lemma}}

\newtheorem{definition}{Definition}

\makeatletter
\newif\if@restonecol
\makeatother

\usepackage[noend]{algpseudocode}
\usepackage[linesnumbered,ruled,vlined, noend]{algorithm2e}
\usepackage[utf8]{inputenc}
\usepackage[english]{babel}
\usepackage[noadjust]{cite}
\usepackage[table]{xcolor}
\usepackage{multirow}

\newcommand{\kwnospace}[1]{{\ensuremath {\mathsf{#1}}}}
\newlength{\figsize} 

\renewcommand{\footnotesize}{\small}

\newcommand{\ggk} {G_{\geq k}}
\newcommand{\gge} {G_{\geq \epsilon}}
\newcommand{\ggei} {G_{\geq \epsilon_{i}}}
\newcommand{\ggeo} {G_{\geq \epsilon_{1}}}

\newcommand{\igeo} {I_{\geq \epsilon_{1}}}
\newcommand{\igei} {I_{\geq \epsilon_{i}}}
\newcommand{\ei} {\epsilon_i}
\newcommand{\eii} {\epsilon_{i+1}}

\newcommand{\nb} {N_G}
\newcommand{\nbp} {N_{G_{\geq \epsilon_{i}}}}
\newcommand{\p} {p}

\newcommand{\degree} {d}
\newcommand{\btf}{\,\mathbin{\resizebox{0.12in}{!}{\rotatebox[origin=c]{90}{$\Join$}}}}

\newcommand{\bs}{\kwnospace{BiT}\textrm{-}\kwnospace{BS}\xspace}

\newcommand{\new} {\kwnospace{BiT}\textrm{-}\kwnospace{BU}\xspace}
\newcommand{\newm} {\kwnospace{BiT}\textrm{-}\kwnospace{BU^{+}}\xspace}
\newcommand{\newa} {\kwnospace{BiT}\textrm{-}\kwnospace{BU^{++}}\xspace}
\newcommand{\newap} {\kwnospace{BiT}\textrm{-}\kwnospace{PC}\xspace}

\newcommand{\bts}{\phi}
\newcommand{\bkts} {$k$-bitruss\xspace}

\newcommand{\btsd} {\textit{bitruss decomposition}\xspace}
\newcommand{\blm} {B}
\newcommand{\kblm} {$k$\textrm{-}B}
\newcommand{\mblm} {B^*}
\newcommand{\re} {r(e)}
\newcommand{\bei} {\kwnospace{BE}\textrm{-}\kwnospace{Index}\xspace}
\newcommand{\nbi} {N_I}
\newcommand{\sib} {twin}

\renewcommand{\baselinestretch}{0.963}

\begin{document}

\title{Efficient Bitruss Decomposition for Large-scale Bipartite Graphs}

\author{
	\IEEEauthorblockN{Kai Wang$^{\dagger}$, Xuemin Lin$^{\dagger}$, Lu Qin$^\star$, Wenjie Zhang$^\dagger$, Ying Zhang$^\star$}
        \vspace{3.6mm}
	\IEEEauthorblockA{
		$^\dagger$University of New South Wales,
        $^\star$University of Technology Sydney\\
		kai.wang@unsw.edu.au, \{lxue,zhangw\}@cse.unsw.edu.au, \{lu.qin, ying.zhang\}@uts.edu.au
	}
}

\maketitle

\begin{abstract}
Cohesive subgraph mining in bipartite graphs becomes a popular research topic recently. An important structure $k$-bitruss is the maximal cohesive subgraph where each edge is contained in at least $k$ butterflies (i.e., $(2, 2)$-bicliques). In this paper, we study the \btsd problem which aims to find all the $k$-bitrusses for $k \geq 0$. The existing bottom-up techniques need to iteratively peel the edges with the lowest butterfly support. In this peeling process, these techniques are time-consuming to enumerate all the supporting butterflies for each edge. To relax this issue, we first propose a novel online index --- the \bei which compresses butterflies into $k$-blooms (i.e., $(2, k)$-bicliques). Based on the \bei, the new bitruss decomposition algorithm \new is proposed, along with two batch-based optimizations, to accomplish the butterfly enumeration of the peeling process in an efficient way. Furthermore, the \newap algorithm is devised which is more efficient against handling the edges with high butterfly supports. We theoretically show that our new algorithms significantly reduce the time complexities of the existing algorithms. Also, we conduct extensive experiments on real datasets and the results demonstrate that our new techniques can speed up the state-of-the-art techniques by up to two orders of magnitude. 

\end{abstract}


\section{Introduction}
\label{sct:introduction}

Bipartite networks are widely used in many real-world applications where we need to model relationships between two different types of entities. For example, author-paper relationships (e.g., authors form the upper layer and papers form the lower layer in the network in Figure \ref{fig:example1}), user-product relationships, etc. Consequently, cohesive subgraph mining in bipartite networks (graphs) becomes a popular research topic recently. In unipartite graphs, there are extensive studies on $k$-truss decomposition  \cite{cohen2008trusses, saito2008extracting, wang2012truss, zhang2012extracting} which constructs the hierarchy of $k$-trusses (each edge in $k$-truss is contained in at least $k$ triangles). However, $k$-truss decomposition cannot be used in bipartite graphs since there is no triangle structure existing in bipartite graphs. Also, since the degree distributions of most real-world bipartite graphs are skewed, it will cause the explosion in the number of edges/triangles if we project bipartite graphs to unipartite graphs \cite{sariyuce2018peeling}. 

In bipartite graphs, butterfly (i.e., a complete $2 \times 2$ biclique) \cite{wang2014rectangle, sanei2018butterfly, wang2019vertex} is the smallest non-trivial cohesive structure and is recognised as an analogue of triangle in unipartite graphs. Based on butterfly, $k$-bitruss is defined as the cohesive subgraph where each edge is contained in at least $k$ butterflies \cite{zou2016bitruss, sariyuce2018peeling}. Consequently, the bitruss number of an edge $e$, denoted by $\bts_e$, is defined as the largest $k$ such that a $k$-bitruss contains $e$. In this paper, we study the {\em bitruss decomposition} problem, which computes the bitruss number for each edge in a bipartite graph. For instance, in Figure \ref{fig:example1}, the bitruss numbers of the edges in blue color (i.e., $(u_0, v_0)$, $(u_0, v_1)$, $(u_1, v_0)$, $(u_1, v_1)$, $(u_2, v_0)$, $(u_2, v_1)$), yellow color (i.e., $(u_2, v_2)$, $(u_3, v_1)$, $(u_3, v_2)$) and gray color (i.e., $(u_2, v_3)$, $(u_3, v_4)$) are 2, 1 and 0, respectively. In the literature, the study of {\em bitruss decomposition} can be easily adopted in many applications. We list some examples below.

\begin{figure}[t]
\begin{centering}
\includegraphics[trim=0 0 0 0,width=0.32\textwidth]{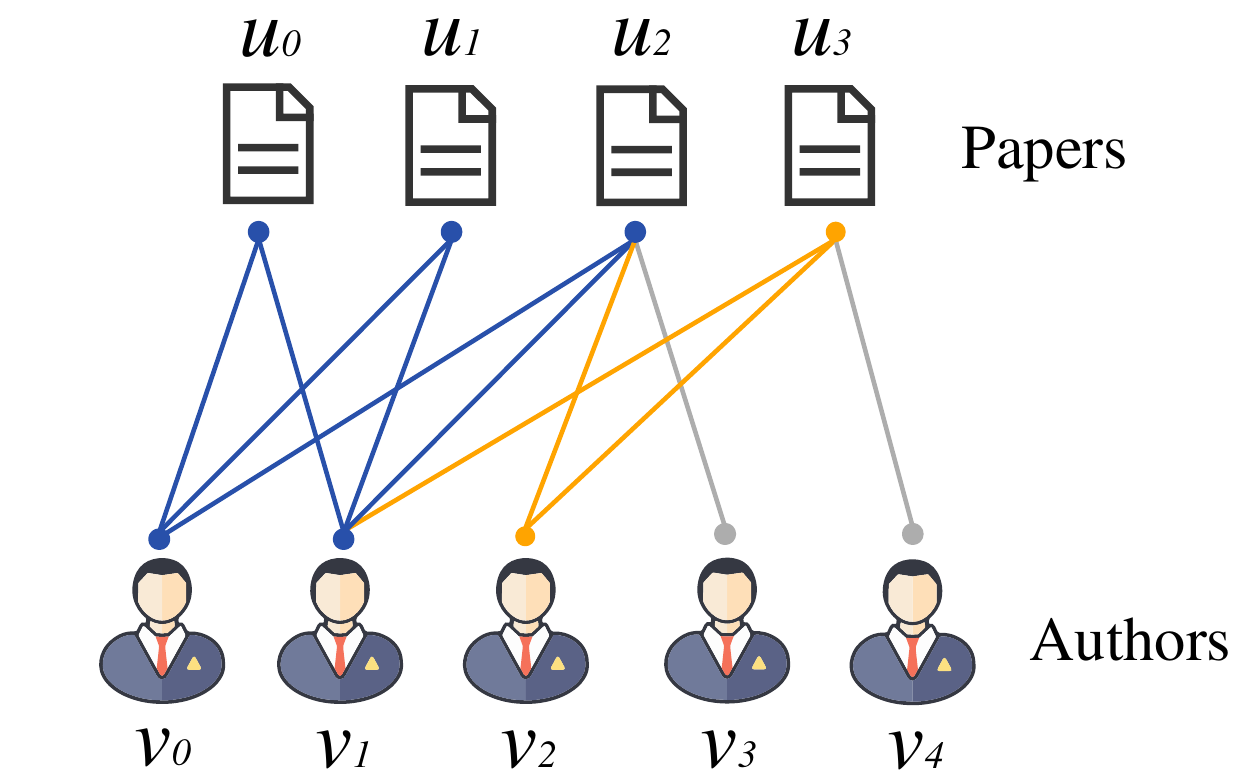}
\vspace*{-2mm}\caption{An author-paper bipartite network}
\label{fig:example1}
\vspace*{-1mm}
\end{centering}
\end{figure}

\noindent
$\bullet$
{\em Fraud detection.} In social media such as Facebook, there exist fraudulent users who give fake ``like''s. Also, with the improvement of the fraud detection techniques, the cost of opening fake accounts is increased, thus frauds cannot rely on too many fake accounts \cite{beutel2013copycatch}. Therefore, these malicious users tend to form a closely connected group. Although the size of the cluster of frauds is unknown, the output of bitruss decomposition applied on the bipartite network (e.g., user-page network) can reveal the close communities at different level of granularity for further investigation.

\noindent
$\bullet$
{\em Identifying nested research groups.} Bipartite graphs are natural fits for modelling the relationship between authors and publications. The bitruss decomposition algorithm can reveal the hierarchical relations of researchers by finding a loose connected research group first and further decomposing it into smaller, more cohesive groups \cite{sariyuce2018peeling}. For instance, in Figure \ref{fig:example1}, all the researchers belong to a loosely research group, while $\{v_0, v_1, v_2\}$ constructs a more cohesive one, and $\{v_0, v_1\}$ constructs the most cohesive research group.

\noindent
$\bullet$
{\em Recommendation system.} When applied to bipartite graphs with user-item structure, bitruss decomposition algorithm can effectively identify dense subgraphs in hierarchical manner. The denser the subgraph is, the more similar the users/items are in this subgraph. Finding users/items at different similarity levels is especially helpful to support the construction of recommendation systems \cite{su2009survey}. 

In real-world applications, the graphs can be very large and the state-of-the-art algorithms cannot handle large-scale bipartite graphs efficiently.
For example, on the graph \texttt{Wiki-it} with $10^7$ edges, the decomposition algorithm in \cite{sariyuce2018peeling} needs more than 30 hours to solve the bitruss decomposition problem as evaluated. Therefore, the study of more efficient bitruss decomposition algorithms is essential to support large-scale graph analysis.

\noindent
{\bf Existing techniques.} \cite{zou2016bitruss, sariyuce2018peeling} both propose a bottom-up approach by iteratively peeling the edges with the lowest butterfly support. It has two key steps: (1) in the counting process, for each edge $e$, it counts the number of butterflies containing $e$ (i.e., the butterfly support of $e$ --- $\btf_e$); (2) in the peeling process, it iteratively removes the edge $e$ with minimum $\btf_e$ and assigns the bitruss number to $e$ as $\btf_e$. To complete the counting process, a novel algorithm recently proposed in \cite{wang2019vertex} takes $O(\sum_{(u, v) \in E(G)} \min\{\degree(u), \degree(v)\})$ time; on the other hand, the peeling process still requires $O(|E(G)|^2)$ time in \cite{zou2016bitruss} or $O(\sum_{(u, v) \in E(G)}\sum_{w \in \nb(v)} \max\{\degree(u), \degree(w)\})$ time in \cite{sariyuce2018peeling}  and, consequently, becomes the performance bottleneck of bitruss decomposition. Here, $E(G)$ denotes the edge set of a graph $G$,  $\degree(v)$ and $\nb(v)$ denote the degree and the neighbor set of a vertex $v$, respectively.

\noindent
{\bf Motivation and challenges.}
In the peeling process, when an edge $e$ is removed, the butterfly supports of the edges which share butterflies with $e$ need to be updated correspondingly. In \cite{zou2016bitruss, sariyuce2018peeling}, this {\em edge removal operation} needs to enumerate all the butterflies containing $e$.
The butterfly enumeration methods used by \cite{zou2016bitruss, sariyuce2018peeling} are inherently the same ---
enumerate the combinations of four vertices with three edges first, then check whether there exists the forth edge to form a butterfly.
The main drawback of the existing combination-based methods is that if the forth edge does not exist (e.g., the butterfly $[u_1, v_1, u_2, v_2]$ does not exist in Figure \ref{fig:existing}(a)), the time of combining and checking is wasted. For instance, considering the graph in Figure \ref{fig:existing}(a) with 4002 vertices, $u_0$ is connected with $v_0$, $v_1$, and $u_1$ ($v_1$) is connected with $v_0$ to $v_{1000}$ ($u_0$ to $u_{1000}$), and $u_2$ ($v_2$) is connected with $v_{1001}$ to $v_{2000}$ ($u_{1001}$ to $u_{2000}$), respectively. When edge $(u_1, v_1)$ is removed, the existing algorithms enumerate butterflies containing $(u_1, v_1)$ by (1) checking whether there is an edge between $u_1$'s neighbors and $v_1$'s neighbors which needs $\degree(u_1) \times \degree(v_1) = 1001 \times 1001$ checks \cite{zou2016bitruss}; or (2) checking whether there is an edge between $v_1$'s two-hop neighbors (e.g., $v_{1001}$) and $u_1$ which needs $\sum_{w \in \nb(v_1)} \max\{\degree(u_1), \degree(w)\} = 1001 \times 1001$ checks \cite{sariyuce2018peeling}. However, there only exists one butterfly containing $(u_1, v_1)$: $[u_0, v_0, u_1, v_1]$.



In addition, we observe that the degree distributions of most real-world graphs are skewed (e.g., \texttt{Wiki-it} and \texttt{Delicious}). In these graphs, some edges can have very high butterfly supports (i.e., {\em hub edges}), though their bitruss numbers are comparatively much smaller. For example, the maximum bitruss number for an edge is only $6{,}638$ on the \texttt{Delicious} dataset, while its butterfly support reaches $1{,}219{,}319$.
For those hub edges, it requires a large number of butterfly support updates to obtain their bitruss numbers in the peeling process.

\begin{figure}[htb]
\begin{centering}
\includegraphics[trim=0 0 0 0,width=0.40\textwidth]{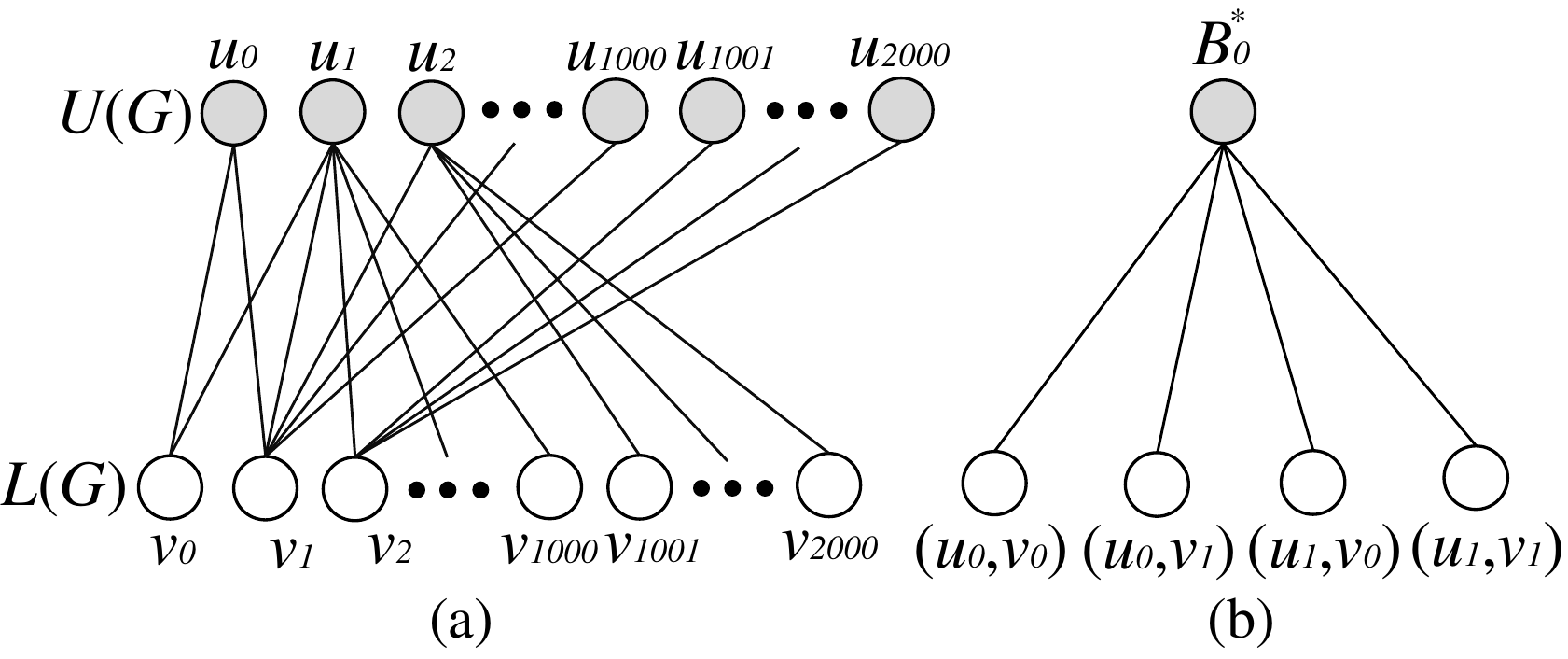}
\vspace*{-1mm}
\caption{Observations}
\label{fig:existing}
\end{centering}
\end{figure}

\begin{figure}[htb]
\begin{centering}
\includegraphics[trim=0 0 0 0,width=0.36\textwidth]{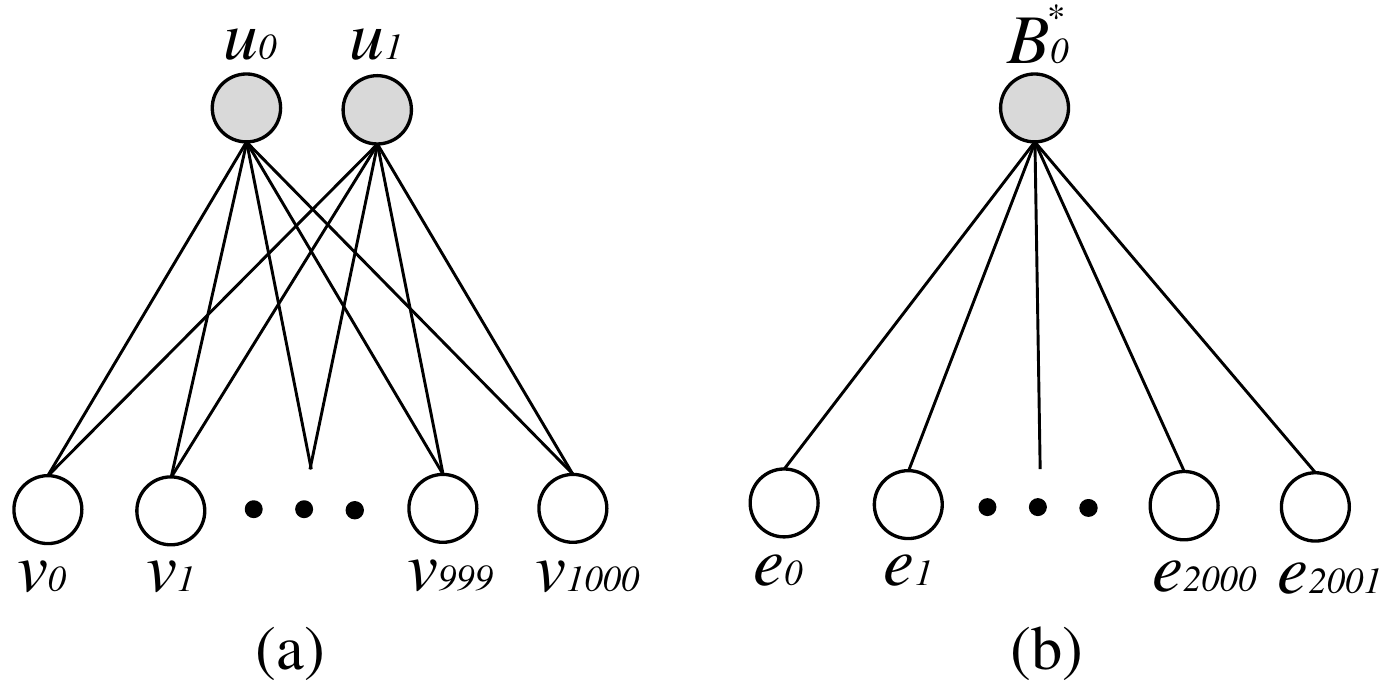}
\vspace*{-1mm}
\caption{(a) a bipartite graph (also a 1001-bloom), (b) the corresponding \bei, $(u_0, v_i)$ is denoted as $e_i$, $(u_1, v_i)$ is denoted as $e_i+1001$}
\label{fig:new}
\end{centering}
\end{figure}

Motivated by the above observations, in this paper, we aim to significantly improve the efficiency of bitruss decomposition by addressing the following two major challenges:
\begin{enumerate}
\item When performing edge removal operations, it is a challenge to efficiently enumerate the butterflies containing each removed edge.
\item It is also a challenge to efficiently handle edges with high butterfly supports (i.e., hub edges).
\end{enumerate}

\noindent
{\bf Our approaches.} To address Challenge 1, we observe that the $bloom$ structure (i.e., a biclique with exactly 2 vertices in one layer) is the combination of butterflies which may have the ability to be used in compacting the butterflies. For example, in Figure \ref{fig:new}(a), the graph is a $1001$-bloom (also a $(2, 1001)$-biclique) which contains $\frac{1001*(1001-1)}{2}$ butterflies. Thus, given a bipartite graph $G$, we can compact all the butterflies in $G$ into blooms. Besides, to guarantee that each butterfly is contained in exactly one bloom, we only identify the maximal priority-obeyed blooms --- the maximal bloom where the vertex with the largest priority belongs to the layer with only two vertices. Here, the higher the degree, the higher the priority; and the ties are broken by vertex ID. Then, the index is constructed by linking the maximal priority-obeyed blooms with the edges they contain; that is the \underline Bloom-\underline Edge-Index (\bei). For example, for the graph in Figure \ref{fig:new}(a), we can construct the corresponding \bei as shown in Figure \ref{fig:new}(b). The \bei can be efficiently constructed after the counting process which needs only $O(\sum_{(u, v) \in E(G)}\min\{\degree(u), \degree(v)\})$ time. Then, when we perform an edge removal operation for $e$, we can directly find all the affected edges through the blooms in \bei rather than enumerating the butterflies containing $e$ using combination-based methods as what existing techniques does. For example, to remove $(u_1, v_1)$ in Figure \ref{fig:existing}(a), we can directly find the 4 edges to be updated in \bei as shown in Figure \ref{fig:existing}(b) instead of using $1001 \times 1001$ butterfly checks in existing solutions. Also as shown in Figure \ref{fig:new}, we can also directly find all the affected edges if one of those edges is removed. Based on \bei, the total peeling process needs only $O(\btf_G)$ time where $\btf_G$ is the number of butterflies in the graph $G$.

To address Challenge 2, we propose the progressive compression approach \newap based on the observation that $\btf_e$ is a lower bound of $\bts_e$ for an edge $e$. Unlike the bottom-up algorithms which process the edges with minimum butterfly supports first, \newap handles a bunch of edges with high butterfly supports (i.e., hub edges) first within cohesive subgraphs and compresses those edges after assigning bitruss numbers to them. In this manner, \newap can significantly reduce the number of butterfly support updates, especially for those hub edges. This is because after assigning the bitruss number for a hub edge, we only need to preserve its support in the \bei and do not need to update its butterfly supports when edges with lower bitruss numbers are removed.

\noindent
{\bf Contribution.}
Our principal contributions are summarized as follows.

\begin{itemize}

\item We propose a novel online index --- the \bei. Based on the \bei, our new bitruss decomposition algorithm \new significantly reduces the time complexities of the existing algorithms as shown in Section \ref{sct:new}. We also propose two batch-based optimizations to further enhance the performance of \new. 

\item To deal with the hub edge issue, we propose the \newap algorithm which processes the hub edges within cohesive subgraphs and compresses the processed edges progressively. In this manner, \newap greatly reduces the number of butterfly support updates for those hub edges.

\item We conduct extensive experiments on real bipartite graphs. The result shows that the proposed algorithm \newap outperforms the state-of-the-art algorithm \cite{sariyuce2018peeling} by up to two orders of magnitude. For instance, the \newap algorithm can solve the bitruss decomposition problem within 20 minutes on \texttt{Wiki-it} dataset with $10^7$ edges, while the state-of-the-art algorithm \cite{sariyuce2018peeling} runs more than 30 hours.

\end{itemize}


\noindent \textbf{Organization.} The rest of the paper is organized as follows. The related work directly follows.
Section~\ref{sct:preliminaries} presents the problem definition.
Section~\ref{sct:benchmark1} introduces the existing algorithms \bs. The \bei is presented in Section~\ref{sct:index}. Section~\ref{sct:algorithms} introduces the \bei-based algorithms including \new, \newa and \newap. Section~\ref{sct:experiment} reports the experimental results. 
Section~\ref{sct:conclusion} concludes the paper.

\noindent
{\bf Related Work.}
In the literature, there are many cohesive subgraph models and recent works on graph decomposition are based on these models \cite{lee2010survey}.

\vspace{0.1cm}
\noindent
{\em Unipartite graphs.} In unipartite networks, many models are defined to capture the cohesiveness of subgraphs such as $k$-core \cite{matula1983smallest, seidman1983network, wang2018efficient}, $k$-truss \cite{cohen2008trusses} and clique \cite{luce1949method}. Furthermore, researchers also study the core decomposition \cite{batagelj2003m, khaouid2015k, cheng2011efficient} and truss decomposition \cite{cohen2008trusses, saito2008extracting, wang2012truss, zhang2012extracting} algorithms. Among those works, truss decomposition is the most similar topic. The reason is that the cohesive structure used in the truss decomposition (i.e., triangle) is the smallest non-trivial clique in unipartite networks, while the cohesive structure used in the bitruss decomposition (i.e., butterfly) is the smallest non-trivial biclique in bipartite networks. However, the structures are different (4-hops' circle vs 3-hops' circle) and the applied networks are different (bipartite network vs unipartite network). Thus, the truss decomposition techniques are not applicable.

\vspace{0.1cm}
\noindent
{\em Bipartite graphs.} In bipartite networks, some studies are conducted towards core-like (e.g., ($\alpha, \beta$)-core \cite{liu2019},  ($p, q$)-
core \cite{cerinvsek2015generalized}, fractional $k$-core \cite{giatsidis2011evaluating}), truss-like (e.g., bitruss \cite{zou2016bitruss, sariyuce2018peeling}), and clique-like (e.g., ($p, q$)-biclique \cite{mitzenmacher2015scalable}, quasi-biclique \cite{sim2009mining}) cohesive structures. Among those works, the core-like and clique-like structures are inherently different from bitruss. For instance, ($\alpha, \beta$)-core \cite{liu2019} is the maximal subgraph where the degree of each vertex in the upper/lower layer is at least $\alpha$/$\beta$; biclique  \cite{sim2009mining} is the maximal complete subgraph. Thus, the techniques in these works cannot be used to solve our problem. In \cite{li2013truss}, the authors project the bipartite graph into a unipartite graph and apply the $k$-truss decomposition algorithm. As we mentioned before, this will cause the explosion of edges/triangles. Thus, the study in this paper aims to improve the recent works in \cite{zou2016bitruss, sariyuce2018peeling} which directly solve the \btsd problem. 

\section{Problem Definition}
\label{sct:preliminaries}

In this section, we formally introduce the notations and definitions. Mathematical notations used throughout this paper are summarized in Table~\ref{tb:notations}.

\begin{table}[htb]
\footnotesize
\caption{The summary of notations}
\vspace{-1mm}
  \centering
    \begin{tabular}{|c|l|}
      \hline
      \cellcolor{gray!25}\textbf{Notation} & \cellcolor{gray!25}\textbf{Definition}             \\ \hline

      $G$   &  a bipartite graph \\ \hline
      $V(G) / E(G)$   &  the vertex/edge set of $G$ \\ \hline
      $U(G),L(G)$   &  a vertex layer of $G$ \\ \hline
      $u, v, w, x$  & a vertex in a bipartite graph \\ \hline
      $(u, v), e$  & an edge in a bipartite graph \\ \hline
      $\blm / \kblm$  & a bloom/$k$-bloom in a bipartite graph \\ \hline
      $\mblm$  & a maximal priority-obeyed bloom \\ \hline
      $(u, v, w)$  & a wedge formed by $u$, $v$, $w$ \\ \hline
      $[u, v, w, x]$  &  a butterfly formed by $u$, $v$, $w$, $x$\\ \hline
      $\degree(u) / \p(u)$   & the degree/priority of $u$ \\ \hline
      $\nb(u)$   & the set of neighbors of $u$ \\ \hline
      $\btf_e$  & the number of butterflies containing $e$ \\ \hline
      $\btf_{\blm}/\btf_G$  & the number of butterflies in $\blm$/$G$ \\ \hline
      $\ggk$   & $\ggk \subseteq G$ where $\btf_e \geq k$ for each $e \in \ggk$ \\ \hline
      $n, m$  & the number of vertices and edges in $G$ ($m > n$) \\ \hline
    \end{tabular}
\label{tb:notations}
\end{table}


Our problem is defined over an undirected bipartite graph $G(V=(U, L), E)$, where $U(G)$ denotes the set of vertices in the upper layer, $L(G)$ denotes the set of vertices in the lower layer, $U(G) \cap L(G) = \emptyset$, $V(G) = U(G) \cup L(G)$ denotes the vertex set, and $E(G) \subseteq U(G) \times L(G)$ denotes the edge set. An edge between two vertices $u$ and $v$ in $G$ is denoted as $(u, v)$ or $(v, u)$. The set of neighbors of a vertex $u$ in $G$ is denoted as $\nb(u) = \{ v\in V(G) \mid (u, v) \in E(G) \} $, and the degree of $u$ is denoted as $\degree(u) = |N_G(u)|$. Each vertex $u$ has a unique id and we assume for every pair of vertices $u \in U(G)$ and $v \in L(G)$, $u.id > v.id$. 


\begin{definition}[Wedge]Given a bipartite graph $G(V, E)$ and vertices $u$, $v$, $w \in V(G)$, a path starting from $u$, going through $v$ and ending at $w$ is called a wedge which is denoted as $(u, v, w)$. For a wedge $(u, v, w)$, we call $u$ the start-vertex, $v$ the middle-vertex and $w$ the end-vertex.
\label{def:wedge}
\end{definition}


%

\begin{definition}[Butterfly]Given a bipartite graph $G$ and four vertices $u, v, w, x \in V(G)$ where $u, w \in U(G)$ and $v, x \in L(G)$, a butterfly induced by the vertices $u, v, w, x$ is a (2,2)-biclique of $G$; that is, $u$ and $w$ are both connected to $v$ and $x$, respectively, by edges $(u, v), (u, x), (w, v), (w, x)\in E(G)$.
\label{def:butterfly}
\end{definition}

\begin{definition}[Bloom/$k$-Bloom]Given a bipartite graph $G(V, E)$, a bloom denoted as $\blm$ is a biclique in $G$  where there are exactly two vertices in $U(B)$ (or $L(B)$). Given a positive integer $k$, a $k$-bloom denoted as $\kblm$  is a $(2, k)$-biclique in $G$; that is, there are two vertices in $U(\kblm)$ (or $L(\kblm)$) connected with $k$ vertices in $L(\kblm)$ (or $U(\kblm)$). For a $k$-bloom, we call $k$ the bloom number. Given a set of vertices  $S \subseteq V(G)$ such that the induced subgraph of $S$ is a bloom, we denote this bloom as $\blm(S)$. 
\label{def:bloom}
\end{definition}

A butterfly induced by vertices $u, v, w, x$ is denoted as $[u, v, w, x]$. We denote the number of butterflies containing an edge $e$ as $\btf_e$, the number of butterflies in a bloom $B$ as $\btf_B$ and the number of butterflies in $G$ as $\btf_G$. Also $\btf_e$ is called the butterfly support of $e$.

\begin{definition}[$k$-bitruss]Given a bipartite graph $G$ and a positive integer $k$, a \bkts denoted as $H_k$ is a maximal subgraph of $G$ where $\btf_e \geq k$ for each edge $e \in H_k$.
\label{def:truss}
\end{definition}


\begin{definition}[Bitruss number]Given a bipartite graph $G$, the bitruss number of an edge $e$ denoted as $\bts(e)$ is the largest $k$ such that a \bkts in $G$ contains $e$.
\label{def:truss_number}
\end{definition}


\noindent
\textbf{Problem Statement. }
Given a bipartite graph $G(V,E)$, our \btsd problem is to compute $\bts(e)$ for each edge $e \in E(G)$.


\vspace{0.2cm}
\begin{figure}[hbt]
\begin{centering}
\includegraphics[trim=0 10 0 15,width=0.48\textwidth]{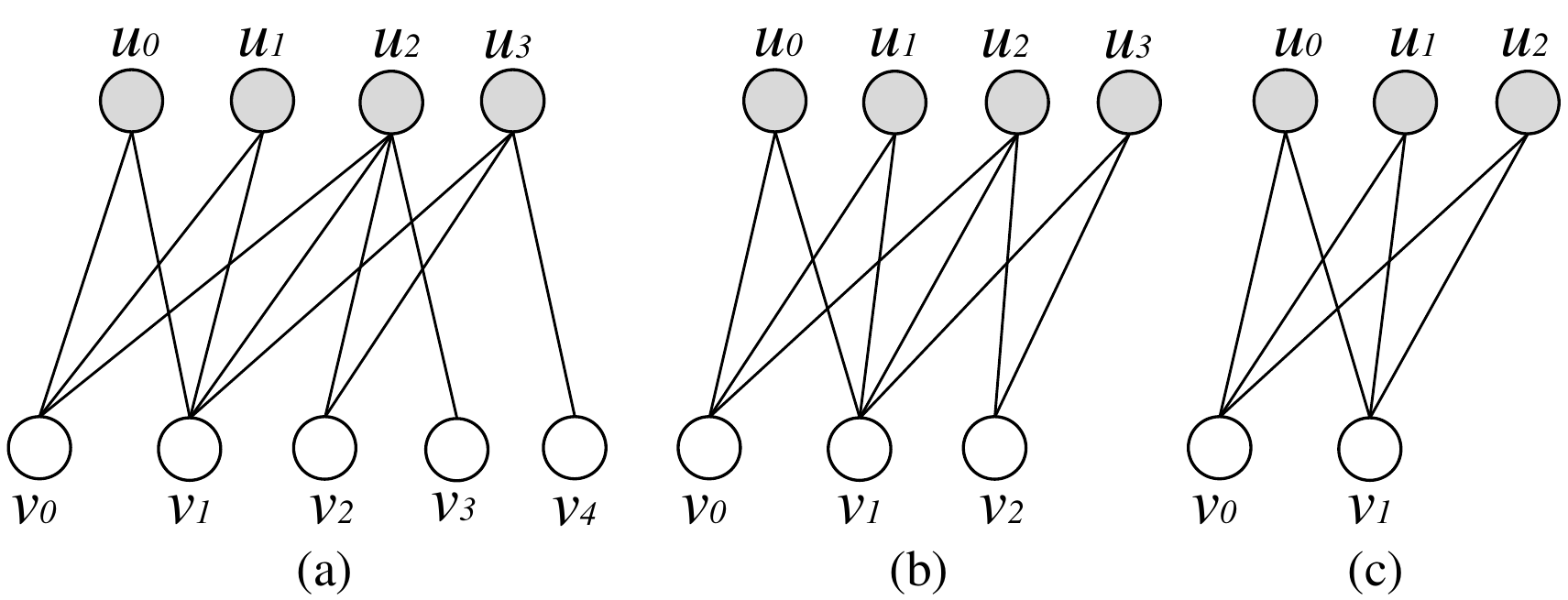}
\caption{(a) the bipartite graph $G$, (b) the 1-bitruss of $G$, $H_1$; (c) the 2-bitruss of $G$, $H_2$}
\label{fig:example}
\vspace*{-1mm}
\end{centering}
\end{figure}

\begin{example}
Considering the bipartite graph $G$ in Figure \ref{fig:example}, the bitruss numbers of the edges in $H_2$ are 2, the bitruss numbers of the edges in $E(H_1) \backslash E(H_2)$ are 1 and the bitruss numbers of the other edges are 0. 
\end{example}


\section{Existing Solutions}
\label{sct:benchmark1}

In this section, we briefly discuss the existing algorithms to solve the $\btsd$ problem. \cite{zou2016bitruss, sariyuce2018peeling} both propose bitruss decomposition algorithms. Since these two algorithms follow the same paradigm with inherently the same peeling idea (as illustrated in the introduction), here we only outline the state-of-the-art algorithm \bs of \cite{sariyuce2018peeling} in Algorithm \ref{algo:benchmark}. 




\begin{algorithm}[hbt]
\small
\DontPrintSemicolon
\KwIn{$G(V = (U, L), E)$: the input bipartite graph}
\KwOut{$\bts_e$ for each $e \in E(G)$}
compute $\btf_e$ for each $e \in E(G)$ // the counting process\;
\ForEach{unassigned $e=(u, v)$ with minimum $\btf_e$} {
    $\bts_e \gets \btf_e$\;
    \ForEach{$w \in \nb(v) \setminus u$} {
        \ForEach{$x \in \nb(w) \cap \nb(u)\setminus v$} {
           \ForEach{$edge\ e' \in [u, v, w, x]$ and $e' \neq e$} {
                \If{$\btf_{e'} > \btf{e}$} {
                    $\btf_{e'} \gets \btf_{e'} - 1$\;
                }
            }
        }
    }
    $E(G) \gets E(G) \backslash e$\;
    mark $e$ as assigned\;
}
\Return{$\bts_e$ for each $e \in E(G)$}
\caption{{\sc \bs}}
\label{algo:benchmark}
\end{algorithm}

As shown in \cite{sariyuce2018peeling}, the time complexity of \bs is $O(\sum_{u \in L(G)}\sum_{v_1, v_2 \in \nb(u)} \max\{\degree(v_1), \degree(v_2)\} + \sum_{(u, v) \in E(G)}\sum_{w \in \nb(v)} \max\{\degree(u), \degree(w)\})$ where the first term is for the counting process and the second term is for the peeling process. The time complexity of the counting process can be reduced to $O(\sum_{(u, v) \in E(G)} \min\{\degree(u), \degree(v)\})$ using the algorithm in \cite{wang2019vertex}. 

\vspace{0.1cm}
\noindent
{\textbf{The performance bottleneck of \bs.}} Here we analyse the dominant cost of \bs. We first define the $edge\ removal\ operation$ as follows.

\begin{definition}[Edge removal operation]Given a bipartite graph $G(V, E)$ and an edge $e \in G$, an edge removal operation for $e$ denoted as $\re$ has two steps. Firstly, find all the edges which share at least one butterfly with $e$ in $G$ and compute their butterfly supports in $G \backslash e$. Secondly, remove $e$ from $G$.
\label{def:ero}
\end{definition}


\vspace*{-2mm}
\begin{figure}[htb]
\begin{centering}
\includegraphics[trim=32 0 0 0,clip,width=0.28\textwidth]{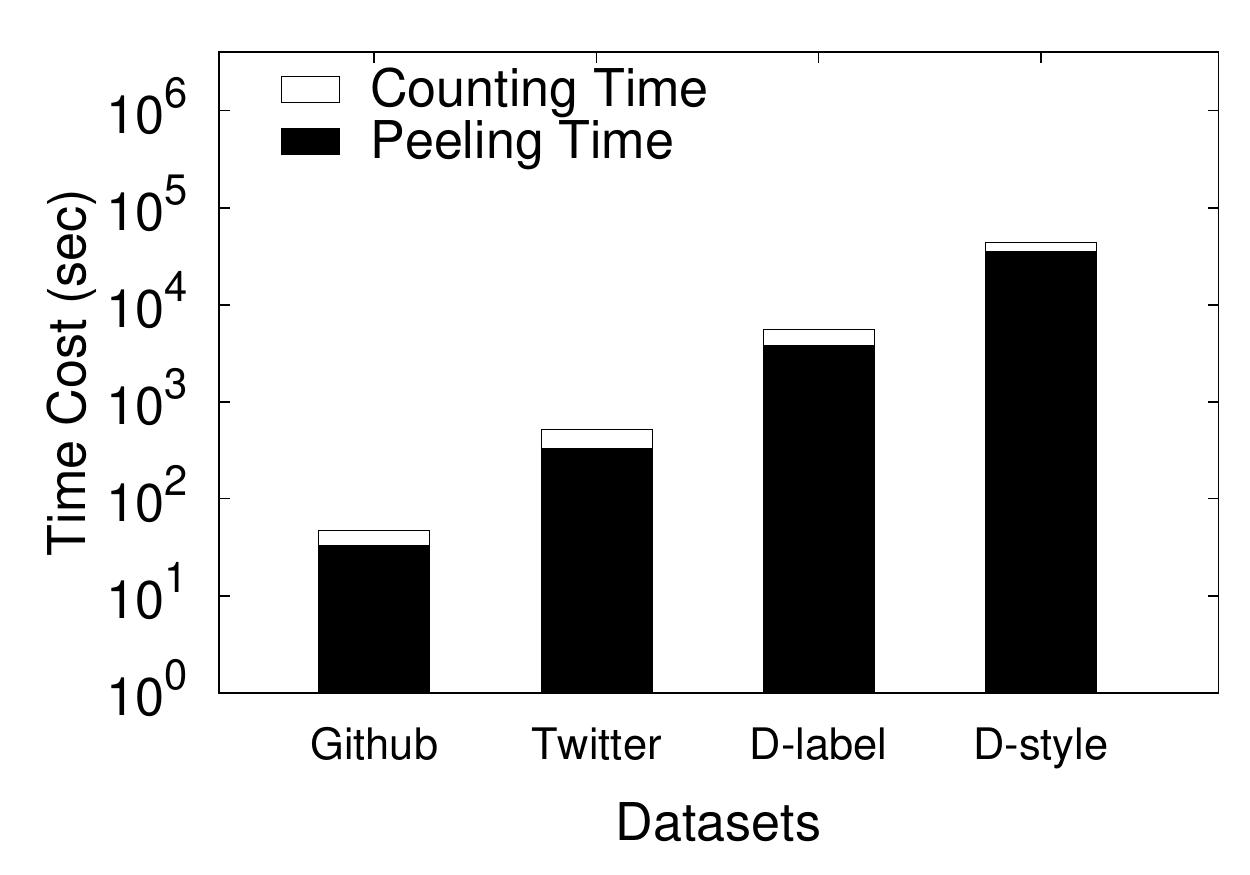}
\vspace*{-2mm}\caption{Time cost of \bs on different datasets}
\label{fig:mot}
\end{centering}
\end{figure}

As shown in Figure \ref{fig:mot}, the dominant cost of \bs is to accomplish the peeling phase on the testing datasets. Moreover, the dominant cost in the peeling phase is incurred when performing the edge removal operations as shown in Algorithm \ref{algo:benchmark}. 


\section{A novel BE-Index}
\label{sct:index}

In this section, we try to explore a compact online index to speed up the edge removal operation.

\subsection{Index Overview}
\vspace{0.1cm}

\noindent
Since a butterfly is a $(2, 2)$-biclique and a $k$-bloom is a $(2, k)$-biclique, we have the following lemma. 

\begin{lemma}
\label{lemma:bloom}
A $k$-bloom contains exactly $\frac{k*(k-1)}{2}$ butterflies.
\end{lemma}

\begin{proof}
\label{proof:bloom}
According to Definition \ref{def:butterfly} and \ref{def:bloom}, this lemma holds.
\end{proof}

For example, as shown in Figure \ref{fig:example}(c), the $3$-bloom $H_2$ contains 3 butterflies $[u_0, v_0, u_1, v_1]$, $[u_0, v_0, u_2, v_1]$, and $[u_1, v_0, u_2, v_1]$. In addition, from the above lemma, we can immediately get the following lemma:

\begin{lemma}
\label{lemma:edge_bloom}
For each edge $e$ contained in a $k$-bloom, there exist $k-1$ butterflies containing $e$.
\end{lemma}

For example, as shown in Figure \ref{fig:example}(b), the edge $(u_2, v_1)$ is contained in a 3-bloom ($H_2$) and a 2-bloom ($[v_1, u_2, v_2, u_3]$), thus there are 2+1 butterflies containing it. Consequently, using blooms instead of butterflies to construct an index should be an effective way to speed up the edge removal operations.

\noindent
{\textbf{The structure of \bei.}} Before introducing a more compact index, we first give the following definitions.

\begin{definition}[Priority]Given a bipartite graph $G(V, E)$, for a vertex $u \in V(G)$, the priority $p(u)$ is an integer where $p(u) \in [1, |V(G)|]$. For two vertices $u, v \in V(G)$, $p(u) > p(v)$ if\\
\vspace*{-3mm}
\begin{itemize}
\item $\degree(u) > \degree(v)$, or \\
\vspace*{-3mm}
\item $\degree(u) = \degree(v)$, $u.id > v.id$.\\
\end{itemize}
\label{def:priority}
\end{definition}
\vspace*{-6mm}

\begin{definition}[Maximal priority-obeyed bloom] Given a bipartite graph $G$, a bloom is a maximal priority-obeyed bloom $\mblm(V(U, L), E)$ if it satisfies the following constricts:
\begin{enumerate}
\item if $v$ has the largest priority in $V(\mblm)$, $v \in U(\mblm) (or\ L(\mblm))$ where $|U(\mblm)| (or\ |L(\mblm)|) = 2$; the layer containing $v$ is called the dominant layer of $\mblm$.
\item there exists no another bloom $B' \supseteq \mblm$ satisfying 1.
\end{enumerate}
\label{def:mblm}
\end{definition}

\begin{lemma}
\label{lemma:maximalbloom}
A butterfly must be contained in one and exactly one maximal priority-obeyed bloom.
\end{lemma}

\begin{proof}
\label{proof:maximalbloom}
According to Definition \ref{def:bloom}, since a butterfly itself is also a bloom, we only need to prove that a butterfly cannot be contained in more than one maximal priority-obeyed bloom. We prove it by contradiction. Suppose we have a butterfly $[u, v, w, x]$ where $u$ has the largest priority in it, $w$ is in the same layer with $u$, and there are two different maximal priority-obeyed blooms $\mblm_1$ and $\mblm_2$ both containing $[u, v, w, x]$. By Definition \ref{def:mblm}, $u$ and $w$ must belong to the dominant layers of $\mblm_1$ and $\mblm_2$ as $u$ has the largest priority. Since $\{u,w\}\times (V(\mblm_1)\backslash \{u,w\}) \in E(\mblm_1)$ and $\{u,w\}\times (V(\mblm_2)\backslash \{u,w\}) \in E(\mblm_2)$, we have $\{u,w\}\times (V(\mblm_1) \cup V(\mblm_2) \backslash \{u,w\}) \in E(\mblm_1) \cup E(\mblm_2)$, i.e., $\blm' = \mblm_1 \cup \mblm_2$ must be a bloom which also satisfies constraint 1 of Definition \ref{def:mblm}. $\blm' \supseteq \mblm_1$ and $\blm' \supseteq \mblm_2$; a contradiction to the constraint 2 of Definition \ref{def:mblm}. Thus, this lemma holds.
\end{proof}

Now we propose the \bei (\underline Bloom-\underline Edge-Index) to speed up the edge removal operation. Given a bipartite graph $G$, a \bei denoted as $I(V(U, L), E)$ links all the maximal priority-obeyed blooms with all the edges in $G$. The structure of the \bei $I$ is summarized as follows:

Each vertex in $U(I)$ corresponds to a maximal priority-obeyed bloom $\mblm$ in $G$ and contains the following information:
\begin{itemize}
\item the id of $\mblm$;
\item $\btf_{\mblm}$\\
\end{itemize}
\vspace*{-3mm}

Each vertex in $L(I)$ corresponds to an edge $e$ in $G$ and contains the following information:
\begin{itemize}
\item the id of $e$;
\item $\btf_e$\\
\end{itemize}
\vspace*{-3mm}

Two vertices in $V(I)$ are linked together if a maximal priority-obeyed bloom $\mblm$ contains an edge $e$ in $G$. We use $\nbi(e)$ to denote the set of maximal priority-obeyed blooms linked to $e$ in $I$, and we use $\nbi(\mblm)$ to denote the set of edges linked to $\mblm$ in $I$.  For each ($\mblm$, $e$) pair in $E(I)$, we also record the $twin\ edge$ of $e$ in $\mblm$ which is defined as follows.

\begin{definition}[Twin edge] Given a maximal priority-obeyed bloom $\mblm$ and an edge $e \in \mblm$, the twin edge of $e$ in $\mblm$ denoted as $\sib(\mblm, e)$ is the edge sharing a vertex $v$ with $e$, where $v$ is in the non-dominant layer of $\mblm$.
\label{def:twin}
\end{definition}

\begin{figure}[hbt]
\begin{centering}
\includegraphics[trim=0 10 0 15,width=0.45\textwidth]{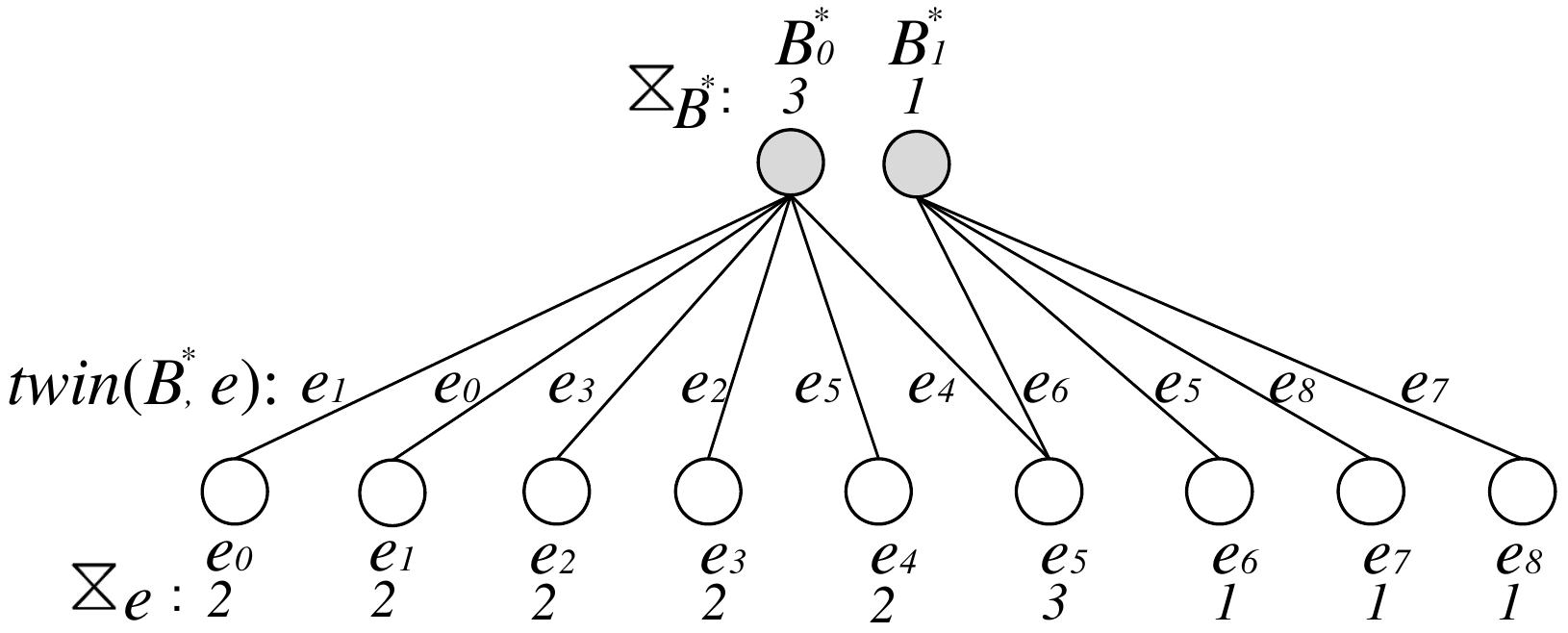}
\caption{The \bei $I$ of $G$ in Figure \ref{fig:example}(a).}
\label{fig:index}
\vspace*{-1mm}
\end{centering}
\end{figure}

\begin{lemma}
\label{lemma:twin}
For each edge $e$ in a maximal priority-obeyed bloom $\mblm$, it has exactly one twin edge in $\mblm$.
\end{lemma}

\begin{proof}
\label{proof:twin}
This lemma immediately follows from Definition \ref{def:bloom} and Definition \ref{def:twin}.
\end{proof}

Now, we give an example of the \bei. From the graph $G$ in Figure \ref{fig:example}(a), we can construct the \bei $I$ as shown in Figure \ref{fig:index}. We denote $(u_0, v_0)$, $(u_0, v_1)$, $(u_1, v_0)$, $(u_1, v_1)$, $(u_2, v_0)$, $(u_2, v_1)$, $(u_2, v_2)$, $(u_3, v_1)$, $(u_3, v_2)$ as $e_0$, $e_1$, $e_2$, $e_3$, $e_4$, $e_5$, $e_6$, $e_7$, $e_8$, respectively. $\mblm_0$ is equal to $H_2$ and $\mblm_1$ is equal to $[u_2, v_1, u_3, v_2]$ in Figure \ref{fig:example}. In $U(I)$, $\btf_{\mblm}$ is recorded and in $L(I)$, $\btf_e$ is recorded (e.g., $\btf_{\mblm_0}$ = 3 and $\btf_{e_0}$ = 2). In $E(I)$, the twin edges are recorded (e.g., the twin edge of $e_0$ in $\mblm_0$ is $e_1$).

\noindent
{\textbf{Perform edge removal operations using \bei.} The key advantage of \bei is that it compresses butterflies into maximal priority-obeyed blooms without losing any butterfly support information. Thus, we can efficiently perform an edge removal operation using \bei as shown in Algorithm \ref{algo:remove}. 

\begin{algorithm}[hbt]
\small
\DontPrintSemicolon
\ForEach{$\mblm \in \nbi(e)$} {
    compute $k$ from $\binom{k}{2} = \btf_{\mblm}$ \;
    \ForEach{$e' \in \nbi(\mblm) \backslash e$} {
        \If{$\btf_{e'} > \btf{e}$} {
            \If{$e'$ = $\sib(\mblm, e)$} {
                $\btf_{e'} \gets \btf_{e'} - (k - 1)$\;
                $E(I) \gets E(I) \backslash (\mblm, e')$\;
            } \Else {
                    $\btf_{e'} \gets \btf_{e'} - 1$\;
            }
        }
    }
     $\btf_{\mblm} \gets \btf_{\mblm} - (k - 1)$\;
}
$E(G) \gets E(G) \backslash e$\;
$L(I) \gets L(I) \backslash e$\;
\caption{{RemoveEdge(e)}}
\label{algo:remove}
\end{algorithm}

Given a bipartite graph $G$, the corresponding \bei $I$ for $G$, and an edge $e \in G$, we first find all the maximal priority-obeyed blooms linked to $e$ in $I$ (i.e., $\nbi(e)$). For each $\mblm \in \nbi(e)$, since it contains $\btf_{\mblm} = \frac{k*(k-1)}{2}$ butterflies if it is a $k$-bloom according to Lemma \ref{lemma:bloom}, we can compute the bloom number $k$ using the equation in line 2. Then, we find the set of edges $\nbi(\mblm)$ for each $\mblm \in \nbi(e)$ (line 3). For each edge $e' \in \nbi(\mblm) \backslash e$, we update the butterfly support $\btf_{e'}$ if $\btf_e' > \btf_e$ (line 4). If $e' = \sib(\mblm, e)$, $e'$ will be contained in no butterfly in $\mblm$ after removing $e$, we decrease $\btf_e'$ by $(k - 1)$ according to Lemma \ref{lemma:edge_bloom} and remove $(\mblm, e')$ from $E(I)$. Otherwise, we decrease $\btf_e'$ by 1. Since $e$ is removed from $\mblm$ and $\mblm$ becomes a ($k$-1)-bloom, we decrease $\btf_{\mblm}$ by $(k - 1)$.

Here is an example of removing an edge with the \bei.

\begin{example}
Consider the bipartite graph $G$ in Figure \ref{fig:example}(a) and the \bei of $G$ in Figure \ref{fig:index}. Suppose we remove the edge $e_6$ as shown in Figure \ref{fig:index}, there are 3 affected edges (i.e., $e_5$, $e_7$ and $e_8$) that can be found through $\mblm_1$ in $I$. Since $e_5$ is the twin edge of $e_6$ in $\mblm_1$ and $\mblm_1$ is a 2-bloom, we need to update $\btf_{e_5}$ to $3 - (2 - 1) = 2$. Then, because the butterfly supports of the edges $e_7$ and $e_8$ are equal to $\btf_{e_6} = 1$, we do not need to update their butterfly supports.
\end{example}

\noindent
{\bf Analysis of the \bei.}} Below, we give some theoretical analysis of the \bei.

\begin{theorem}
\label{theorem:newer}
Given a bipartite graph $G$, the corresponding \bei $I$ for $G$, and an edge $e \in G$, Algorithm \ref{algo:remove} correctly performs an edge removal operation for $e$ using $I$.
\end{theorem}

\begin{proof}
According to Definition \ref{def:ero}, here we only need to prove that the butterfly supports of all the affected edges (i.e., the edges which share butterflies with $e$) are correctly updated. Firstly, we retrieve a set of maximal priority-obeyed blooms containing $e$ from the \bei; it is obvious from Lemma \ref{lemma:maximalbloom} that all the affected edges are contained by these blooms. Secondly, according to Lemma \ref{lemma:bloom} and Lemma \ref{lemma:twin}, the butterfly supports of the affected edges are correctly updated as in Algorithm 2 lines 5 - 9. Thus, this theorem holds.
\end{proof}

\begin{lemma}
\label{lemma:newertc}
Given a bipartite graph $G$ and $e \in G$, it needs $O(\btf_e)$ time to perform Algorithm \ref{algo:remove} for $e$.
\end{lemma}

\begin{proof}
Since there are $O(\btf_e)$ butterflies associated with $e$, the number of affected edges is $O(\btf_e)$. By using index $I$, it takes constant time to assess and update an affected edge. Consequently, the overall time for Algorithm  \ref{algo:remove} is $O(\btf_e)$.
\end{proof}


Before introducing Lemma \ref{lemma:newersc}, we give the below definition:

\begin{definition}[Priority-obeyed wedge] A priority-obeyed wedge is a wedge where the priority of start-vertex is larger than the priorities of middle-vertex and end-vertex.
\label{def:pwedge}
\end{definition}

\begin{lemma}
\label{lemma:newersc}
For a bipartite graph $G$, storing the \bei $I$ of $G$ needs $O(\sum_{(u, v) \in E(G)}\min\{\degree(u), \degree(v)\})$ space.
\end{lemma}

\begin{proof}
The \bei compresses butterflies  into maximal priority-obeyed blooms; the space usage is dominated by the summed number of edges in these maximal priority-obeyed blooms.
Within one maximal priority-obeyed bloom, each edge is contained by exactly one priority-obeyed wedge according to Definition \ref{def:mblm} and Definition \ref{def:pwedge}. Also, can be proved in a similar way as Lemma \ref{lemma:maximalbloom}, one priority-obeyed wedge exists in at most one maximal priority-obeyed blooms; we equivalently prove that the total number of priority-obeyed wedges is bounded by $O(\sum_{(u, v) \in E(G)}\min\{\degree(u), \degree(v)\})$.

Considering an edge $(u, v) \in E(G)$ with $p(u) > p(v)$ (i.e., $d(u) \geq d(v)$), $u$ should be the start-vertex for all  priority-obeyed wedges containing $(u, v)$, and the number of such wedges is $O(\degree(v))$ since there are at most $\degree(v)$ end-vertices linking with the middle-vertex $v$. Consequently, the total number of priority-obeyed wedges in $G$ is $O(\sum_{(u, v) \in E(G), p(u) > p(v)} \degree(v))$ = $O(\sum_{(u, v) \in E(G)}\min\{\degree(u), \degree(v)\})$, this theorem holds.
\end{proof}


\subsection{Index Construction}

\begin{algorithm}[th]
\small
\DontPrintSemicolon
\KwIn{$G(V = (U, L), E)$: the input bipartite graph}
\KwOut{$I$: the \bei} 
//$\btf_e$ for each $e \in E(G)$ is pre-computed\;
Compute $p(u)$ for each $u \in V(G)$  // Definition \ref{def:priority} \;
\ForEach{ $u \in V(G)$ } {
    initialize hashmap $count\_wedge$ with zero\;
    \ForEach{$v \in \nb(u): p(v) < p(u)$} {
        \ForEach{$w \in \nb(v): p(w) < p(u)$} {
            $count\_wedge(w) \gets count\_wedge(w) + 1$\;
        }
    }
    \ForEach{$v \in \nb(u): p(v) < p(u)$} {
        \ForEach{$w \in \nb(v): p(w) < p(u)$} {
            \If {$count\_wedge(w) > 1$} {
                $\mblm \gets$ the bloom anchored by $u$ and $w$\;
                \If {$\mblm \notin U(I)$} {
                    $\btf_{\mblm} \gets \binom{count\_wedge(w)}{2}$\;
                     add $\mblm.id$ and $\btf_{\mblm}$ into $U(I)$\;
                }
                \If {$e = (u, v) \notin L(I)$} {
                    add\ $e.id$ and $\btf_e$ into $L(I)$\;
                }
                \If {$e = (v, w) \notin L(I)$} {
                    add\ $e.id$ and $\btf_e$ into\ $L(I)$\;
                }
                link $\mblm$ with $(u, v)$ in $E(I)$\;
                link $\mblm$ with $(v, w)$ in $E(I)$\;
                $\sib(\mblm, (u, v)) \gets (v, w)$\;
                $\sib(\mblm, (v, w)) \gets (u, v)$\;
            }
        }
    }
}
\Return{$I$}\; 
\caption{{IndexConstruction}}
\label{algo:ic}
\end{algorithm}

To build the \bei, the key step is to get all the maximal priority-obeyed blooms. We have the following lemma.

\begin{lemma}
\label{lemma:bloom_wedge}
A maximal priority-obeyed bloom with the bloom number equal to $k$ must be the combination of $k$ priority-obeyed wedges.
\end{lemma}

\begin{proof}
\label{proof:bloom_wedge}
This lemma immediately follows from Definition \ref{def:wedge}, Definition \ref{def:bloom} and Definition \ref{def:mblm}.
\end{proof}

For example, the bloom in Figure \ref{fig:example}(c) is combined by the priority-obeyed wedges ($v_1, u_0, v_0$), ($v_1, u_1, v_0$) and ($v_1, u_2, v_0$). Based on the above observations, we propose the Index Construction algorithm as shown in Algorithm \ref{algo:ic}. Given a bipartite graph $G$, we first assign a priority to each vertex $u \in V(G)$. After that, we process the wedges from each vertex $u \in V(G)$ and initialize the hashmap $count\_wedge$ with zero. For each $v \in \nb(u)$, we process $v$ if $p(v) < p(u)$ according to Definition \ref{def:mblm}. Then, to avoid duplicate visiting, we only process $w \in \nb(v)$ with $p(w) < p(u)$. After running lines 4 - 7, we get $|\nb(u) \cap \nb(w)|$ (i.e., $count\_wedge(w)$) for $u$ and $w$. According to Definition \ref{def:bloom}, if $count\_wedge(w) > 1$, it means that there is a maximal priority-obeyed bloom $\mblm$ contains the vertices $u$ and $w$ in the dominant layer of $\mblm$. Then, if $\mblm \notin U(I)$, we compute $\btf_{\mblm}$, put $\mblm.id$ and $\btf_{\mblm}$ into $U(I)$ (lines 8 - 14). After that, if an edge $e \in \mblm \notin L(I)$, we put $e.id$ and $\btf_e$ into $L(I)$. Also, in $E(I)$, we link $e$ with $\mblm$ and record $\sib(\mblm, e)$.




\noindent
{\bf Time complexity of constructing the \bei.} The time complexity of constructing the \bei (i.e., Algorithm \ref{algo:ic}) is bounded by the time complexity to find all the maximal priority-obeyed blooms and the edges in them. According to Lemma \ref{lemma:bloom_wedge}, this can be done by finding all the priority-obeyed wedges. Since finding the priority-obeyed wedges in $G$ needs $O(\sum_{(u, v) \in E(G)}\min\{\degree(u), \degree(v)\})$ time which can be proved similarly as the proof of Lemma \ref{lemma:newersc}, this theorem holds.

\section{Decomposition Algorithms}
\label{sct:algorithms}

In this section, we present three \bei-based bitruss decomposition algorithms. We first present a bottom-up approach \new which starts the peeling process from the smallest $k$. Then, we present the algorithm \newa which uses two batch-based optimizations to speed up \new. After that, a progressive compression approach \newap is proposed.

\subsection{A bottom-up approach}
\label{sct:new}
We firstly introduce the \new algorithm using the \bei which are shown in Algorithm \ref{algo:new}. 

\begin{algorithm}[hbt]
\small
\DontPrintSemicolon
\KwIn{$G(V = (U, L), E)$: the input bipartite graph}
\KwOut{$\bts_e$ for each $e \in E(G)$}
compute $\btf_e$ for each $e \in E(G)$\;
call IndexConstruction // Algorithm \ref{algo:ic}\;
$k \gets 0$\;
\While{exist unassigned edges in $G$}{
\While{exist unassigned $e=(u, v)$ with $\btf_e \leq k$} {
    $\bts_e \gets k$\;
    call RemoveEdge($e$) // Algorithm \ref{algo:remove}\;
    mark $e$ as assigned\;
} 
$k \gets k+1$\;
}

\Return{$\bts_e$ for each $e \in E(G)$}
\caption{{\sc \new}}
\label{algo:new}
\end{algorithm}

Given a bipartite graph $G$, \new first computes $\btf_e$ for each edge $e \in E(G)$ (i.e., the counting phase) using the algorithm in \cite{wang2019vertex}. Then, it calls Algorithm \ref{algo:ic} to construct the \bei $I$. After that, in the peeling phase, \new iteratively removes an unassigned edge $e$ with $\btf_e$ less than the current $k$ value and $\bts_e$ of $e$ is assigned as $k$.

\noindent
{\textbf{Analysis of the \new algorithm.}}
Below we show the correctness and time/space complexities of \new.

\noindent
\begin{theorem}
\label{theorem:newco}
The \new algorithm correctly solves the \btsd problem.
\end{theorem}

\begin{proof}
\label{proof:newc}
This theorem directly follows from Theorem \ref{theorem:newer}.
\end{proof}

\noindent
{\bf Time complexity.} \new has three parts. The counting part needs $O(\sum_{(u, v) \in E(G)}\min\{\degree(u), \degree(v)\})$ \cite{wang2019vertex}, the index construction part also needs $O(\sum_{(u, v) \in E(G)}\min\{\degree(u), \degree(v)\})$ time as illustrated before. For the peeling part, since the removing of an edge $e$ needs $O(\btf_e)$ time as proved in Lemma \ref{lemma:newertc} and we need to remove all the edges in $G$. Thus, the time complexity of the peeling process is $O(\sum_{e \in E(G)}\btf_e) = O(\btf_G)$. In total, the time complexity of \new is $O(\sum_{(u, v) \in E(G)}\min\{\degree(u), \degree(v)\} + \btf_G)$.

%
\noindent
\begin{lemma}
\label{lemma:tmc}
Given a bipartite graph $G(V, E)$, we have the following equations:\\
\begin{align}
&\btf_G \le m^2\\
&\btf_G \le \sum_{(u, v) \in E(G)} \sum_{w \in \nb(v) \backslash u} \max\{\degree(u), \degree(w)\}
\end{align}
\end{lemma}

\begin{proof}
\label{proof:newc}
Apparently, for each edge $(u, v) \in E(G)$, there can be at most $E(G)-1$ butterflies that contain $(u, v)$, as the edge $(w, x)$ of a butterfly $[u, v, w, x]$ cannot be shared with any other butterflies containing $(u, v)$. Thus, we can get that $\btf_G \le \sum_{(u, v) \in E(G)} m-1 \le m^2$; the first equation holds. For the second equation above, we alternatively prove a stricter equation. Since $\btf_G = \sum_{(u, v) \in E(G)} \btf_{(u, v)} / 4$ (one butterfly is a $(2, 2)$-biclique containing 4 edges), we will show $\forall (u, v) \in E(G)$, $\btf_{(u, v)} \le \sum_{w \in \nb(v)\backslash u} \max\{\degree(u), \degree(w)\}$. For each edge $(u, v) \in E(G)$, there can be at most $(\degree(u)-1) \times (\degree(v)-1)$ butterflies that contain $(u, v)$. This is because there should exist a vertex $w \in \nb(v) \backslash u$ and a vertex $x \in \nb(u) \backslash v$ to form a butterfly with $u$ and $v$. Thus, we get the equation $\btf_{(u, v)} \leq (\degree(u)-1) \times (\degree(v)-1)$. For each edge $(u, v) \in E(G)$: (1) if $\degree(v)=1$, $\btf_{(u, v)} = \sum_{w \in \nb(v)\backslash u} \max\{\degree(u), \degree(w)\}=0$; (2) if $\degree(v)>1$, $\btf_{(u, v)} \leq (\degree(u)-1) \times (\degree(v)-1)  \leq \sum_{w \in \nb(v)\backslash u} \degree(u) \leq \sum_{w \in \nb(v)\backslash u} \max\{\degree(u), \degree(w)\}$. Thus, this lemma holds.
\end{proof}

From the above lemma, we can get that \new reduces the time complexities of the existing algorithms \cite{zou2016bitruss, sariyuce2018peeling}.


\noindent
{\bf Space complexity.} In the \new algorithm, we need $O(\sum_{(u, v) \in E(G)}\min\{\degree(u), \degree(v)\})$ space to store the \bei as proved in Lemma \ref{lemma:newersc} and $O(m)$ space to store the butterfly supports and bitruss numbers for edges. Thus, the space complexity is $O(\sum_{(u, v) \in E(G)}\min\{\degree(u), \degree(v)\})$.

%

\subsection{Batch-based optimizations}
\label{sct:newa}
Here, we introduce two batch-based optimizations to further improve \new. 

\noindent
{\textbf {Batch edge processing.}} The batch edge processing optimization is based on the following lemma.

\begin{lemma}
\label{lemma:batch1}
In \new, the removing of an edge $e$ does not change $\bts_{e'}$ if $\btf_e = \btf_{e'}$.
\end{lemma}

\begin{proof}
\label{proof:batch1}
This lemma is immediate.
\end{proof}

Lemma \ref{lemma:batch1} is immediate since we only update $\btf_{e'}$ if $\btf_{e'} > \btf_e$ in Algorithm \ref{algo:remove} (lines 7 - 8). Based on Lemma \ref{lemma:batch1}, we can process a set $C$ of edges which contains all the edges with the same butterfly supports in each iteration of peeling. Then, we can compute the total butterfly supports for each edge affected by the removing of edges in $C$ and update the butterfly supports for each affected edge in one step. Thus, the number of butterfly support updates can be reduced. The details of this optimization are shown in Algorithm \ref{algo:newa} later.

\noindent
{\textbf {Batch bloom processing.}}  According to Algorithm \ref{algo:remove}, when removing an edge $e$, we need to go through $\mblm$ to get the affected edge $e'$. Using the batch edge processing strategy, it may need to go through the same bloom many times in \bei. Thus, we consider also processing the blooms in batch. We use an array to record the number of accesses for each bloom in \bei. Then, we process all the accessed $\mblm$ and update butterflies counts for the affected edges.

The details of the algorithm \newa which utilizes the above two strategies are shown in Algorithm \ref{algo:newa}.  Given a bipartite graph $G$, \newa first computes $\btf_e$ for each edge $e \in E(G)$ and constructs the \bei $I$. Then, in the peeling phase, \newa first puts all the unassigned edges with minimum butterfly supports into a set $S$ and initializes {\em MBS} to record the minimum butterfly support in this iteration. We also initialize $C(\mblm)$ for each $\mblm \in U(I)$ to record the number of edge-pairs removed (i.e., a removed edge and its twin edge) of $\mblm$ in one iteration (lines 3 - 5). Then, for each $e \in S$, we assign $\bts_e$ and for each $\mblm \in \nbi(e)$, we increase $C(\mblm)$ by 1 and remove $e' = \sib(\mblm, e)$ from $I$ if $e'$ is not assigned (lines 8 - 13). This is because when an edge is removed from $\mblm$, its twin edge also loses all the supports from $\mblm$ and a pair of twin edges should only count once. Next, if $C(\mblm) > 0$, we also need to update $\btf_{\mblm}$ and $\btf_{e'}$ for each unassigned $e' \in \nbi(\mblm)$ according to Lemma \ref{lemma:bloom}. Then, we mark all the edges in $S$ as assigned and remove them (lines 19 - 21).

\begin{example}
Consider the bipartite graph $G$ in Figure \ref{fig:example}(a) and the \bei of $G$ in Figure \ref{fig:index}. Using the batch-based optimizations, \newa firstly processes all the edges with butterfly support equal to 1 (i.e., $e_6$ to $e_8$ in Figure \ref{fig:index}). Since $\mblm_1$ is a 2-bloom and $e_5$ is the twin edge of $e_6$ in $\mblm_1$, $\btf_{e_5}$ becomes $3-(2-1)=2$ as shown in Algorithm 5 lines 11 - 13. Then, since no other unassigned edges are affected, we only need to update $\btf_{\mblm_1}$ to 0 and assign the bitruss numbers of $e_6$, $e_7$ and $e_8$ as 1. Similarly, in the next peeling iteration, we can process $e_0$ to $e_5$ together. Since they form three pairs of twin edges, we just update $\btf_{\mblm_0}$ to 0 and assign the bitruss numbers of them as 2.
\end{example}

Note that, the worst case time and space complexities of \newa are the same as \new since the batch-based strategies are used to find potential cost-sharing.

%
%
%

\begin{algorithm}[hbt]
\small
\DontPrintSemicolon
\KwIn{$G(V = (U, L), E)$: the input bipartite graph}
\KwOut{$\bts_e$ for each $e \in E(G)$}
run Algorithm \ref{algo:new} lines 1 - 2 \;
\While{exist unassigned edges in $G$} {
    {\em MBS}$ \gets$ minimum butterfly support in this iteration\;
    $S \gets$ unassigned edges with minimum butterfly support \;
    $C(\mblm) \gets 0$ for each $\mblm \in U(I)$\;
    \ForEach{$e \in S$} {
        $\bts_e \gets \btf_e$\;
        \ForEach{$\mblm \in \nbi(e)$} {
            $C(\mblm)$++\;
            compute $k$ from $\binom{k}{2} = \btf_{\mblm}$ \;
            \If{$e'$ = $\sib(\mblm, e)$ is not assigned} {
                $\btf_{e'} \gets max (${\em MBS}$, \btf_{e'} - (k-1))$\;
                $E(I) \gets E(I) \backslash (\mblm, e')$\;
            }
        }
    }
    \ForEach{$\mblm$ : $C(\mblm) > 0$} {
        compute $k$ from $\binom{k}{2} = \btf_{\mblm}$ \;
        $\btf_{\mblm} \gets \frac{(k - C(\mblm))((k-1) - C(\mblm))}{2}$\;
        \ForEach{$e' \in \nbi(\mblm)\ and\ e' \notin S$} {
            $\btf_{e'} \gets max(${\em MBS}$, \btf_{e'} - C(\mblm))$\;
        }
    }
    \ForEach{$e \in S$} {
        $E(G) \gets E(G) \backslash e$, $L(I) \gets L(I) \backslash e$\;
        mark $e$ as assigned\;
    }
}
\Return{$\bts_e$ for each $e \in E(G)$}
\caption{{\sc \newa}}
\label{algo:newa}
\end{algorithm}


\vspace{-0.1cm}
\subsection{A progressive compression approach}
\label{sct:newap}

\begin{figure}[hbt]
\begin{centering}
\includegraphics[trim=0 10 0 15,width=0.46\textwidth]{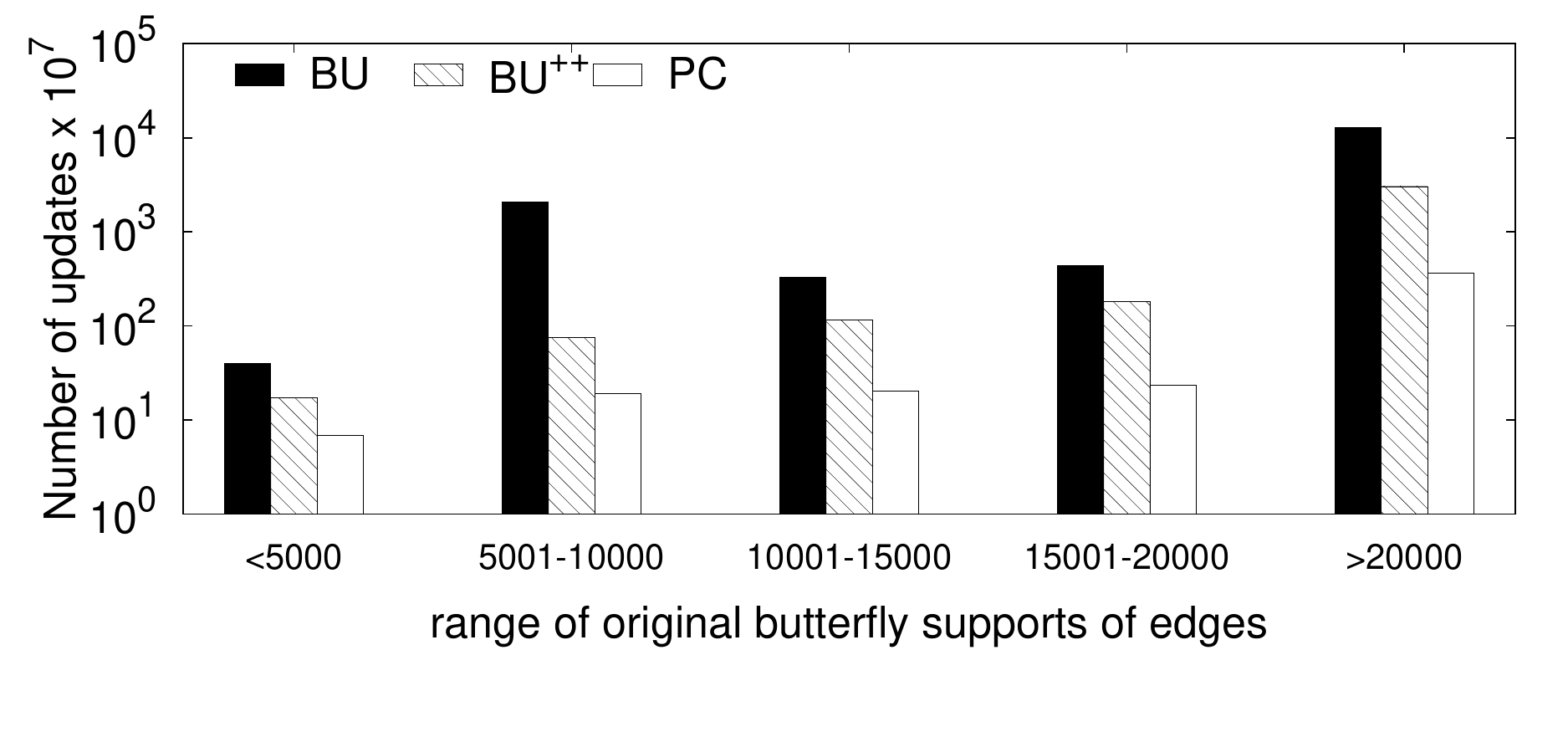}
\vspace*{-4mm}\caption{Number of butterfly support updates to obtain bitruss numbers on \texttt{D-style}.}
\label{fig:motivation_update}
\vspace*{-1mm}
\end{centering}
\end{figure}

\noindent
{\textbf{Motivation.}}
As discussed above, the algorithm \newa using batch-based optimizations already reduces the number of butterfly support updates comparing with \new. However, in \newa, we still cost lots of time to update butterfly supports for those edges with high butterfly supports (i.e., $hub\ edges$). For example, as shown in Figure \ref{fig:motivation_update}, about 80\% update operations are performed for hub edges (i.e., edges with original butterfly supports $> 20,000$) in \newa. To solve this issue, we have the observation of the following lemma. 

\begin{lemma}
\label{lemma:kbitruss}
Given a bipartite graph $G$, the $k$-bitruss of $G$ is contained in a subgraph of $G$ denoted as $\ggk$ where for each edge $e \in \ggk$, $\btf_e \geq k$.
\end{lemma}

\begin{proof}
\label{proof:kbitruss}
This lemma directly follows from Definition \ref{def:truss}.
\end{proof}

Based on Lemma \ref{lemma:kbitruss}, in this section, we introduce a progressive compression approach \newap which aims to reduce the number of butterfly support updates for hub edges.

\begin{figure}[htb]
\begin{centering}
\includegraphics[trim=0 0 0 0,width=0.34\textwidth]{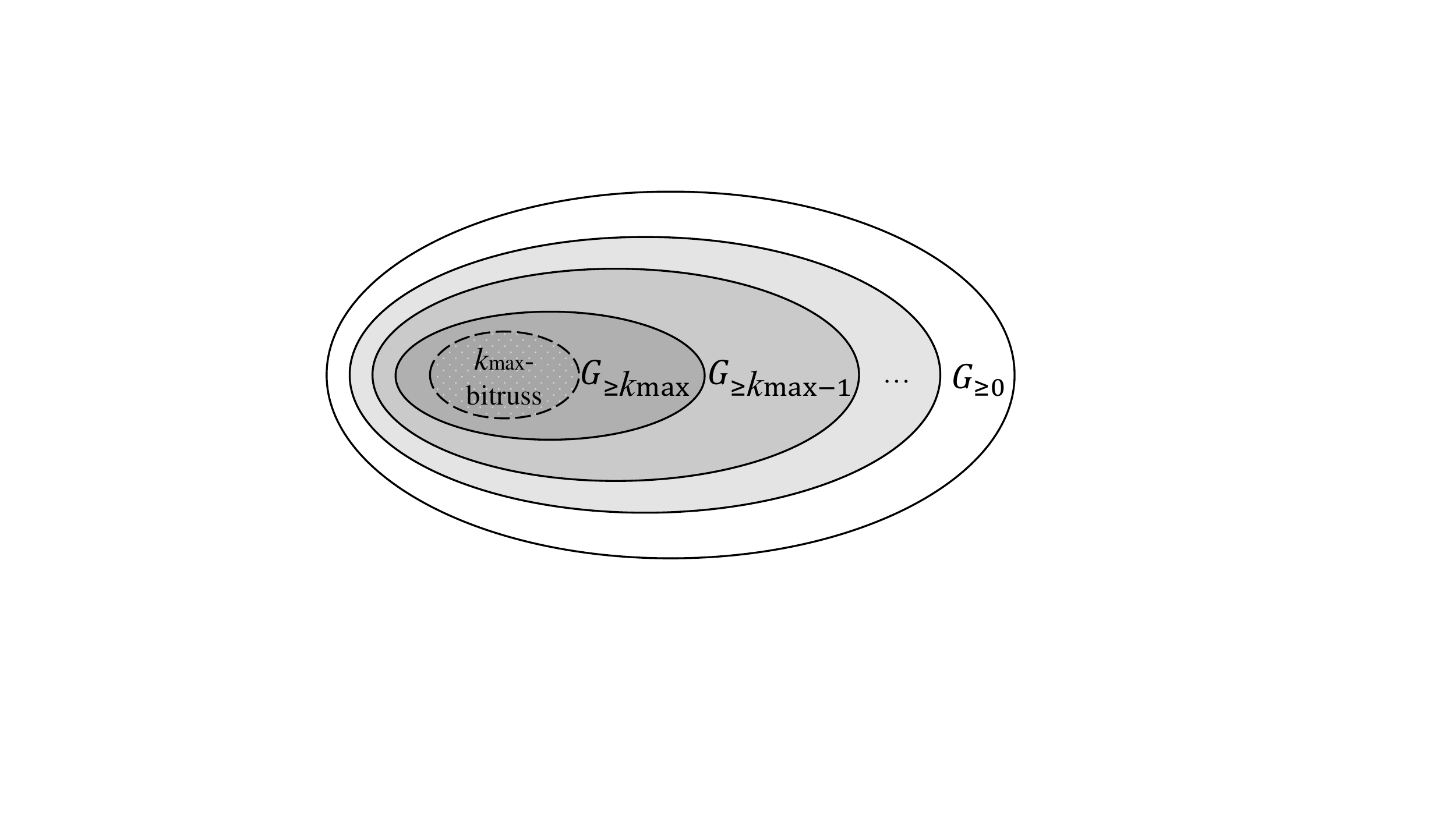}
\caption{Illustrating the \newap algorithm}
\label{fig:progress}
\vspace*{-1mm}
\end{centering}
\end{figure}

\noindent
{\textbf{The algorithmic framework.}} We first introduce the algorithmic framework of \newap. Given a bipartite graph $G$ and the parameter $\epsilon$ (i.e., the butterfly supports threshold in an iteration), \newap has the follows steps: 
\begin{enumerate}
\item extract the candidate subgraph $\gge$;
\item peel $\gge$ similarly as \newa to obtain $\epsilon$-bitruss;
\item decrease $\epsilon$, repeat steps 1 - 2 until $\epsilon = 0$.
\end{enumerate}

According to the above framework, in each iteration of \newap, we only handle the edges with butterfly supports $\geq \epsilon$ and get the bitruss numbers of the edges in $\epsilon$-bitruss. For example, as shown in Figure \ref{fig:progress}, we first consider the candidate graph $G_{\geq {k_{max}}}$ where $k_{max}$ is the largest possible bitruss number. We can get the $k_{max}$-bituss from $G_{\geq {k_{max}}}$ by only considering the edges in $E(G_{\geq {k_{max}}})$. In this manner, the bitruss number of hub edges can be computed within a cohesive subgraph of $G$, and we can avoid performing unnecessary updates (caused by the edges with low butterfly supports). The details of \newap are shown in Algorithm \ref{algo:newap}. 

\noindent
{\textbf{Step 1: candidate subgraph generation.}}
Given a bipartite graph $G$, to generate the first candidate subgraph $\ggeo$, we need to compute $\epsilon_1$ which equals to the largest possible bitruss number $k_{max}$ in $G$. We set $k_{max}$ as the largest integer if there exists at least $k_{max}$ edges in $G$ with their butterfly supports $\geq k_{max}$. It can be easily computed after sorting the edges in non-ascending order of their butterfly supports. In other iteration with $i > 1$, $\ei$ is computed in Step 3. In iteration $i$, we extract the candidate subgraph $\ggei$ where $\btf_e \geq \ei$ for each edge $e$ in $\ggei$. Then, we recompute $\btf_e$ for each edge $e$ on $\ggei$ and remove $e$ from $\ggei$ if $\btf_e < \ei$ (lines 5 - 6). 


\noindent
{\textbf{Step 2: compressed index construction and index-based computation.}}
According to Lemma \ref{lemma:kbitruss}, we can compute the $\epsilon$-bitruss $H_\epsilon$ on $\gge$ and obtain the bitruss numbers of all the edges in $H_\epsilon$. Following the similar idea as \newa, in iteration $i$, we construct the \bei $\igei$ based on $\ggei$ and run the peeling process. Note that, for the first iteration with $\epsilon_1 = k_{max}$, we just construct the \bei $\igeo$ based on $\ggeo$ and run the peeling process similar as \newa. For the other iterations with $i > 1$, since there may exist assigned edges (i.e., edges with their bitruss numbers assigned in previous iterations) in the candidate graph $\ggei$, we do not insert these assigned edges into $\igei$ but preserve the blooms they supported in $\igei$. As shown in Algorithm \ref{algo:ica} lines 8 - 14, we only insert the unassigned edges into $\igei$, but the blooms are preserved in $\igei$. In this manner, (1) we can get the correct butterfly supports of unassigned edges in $\ggei$; (2) when removing an unassigned edge $e$ in $\ggei$, we do not update the supports of the assigned edges which share butterflies with $e$. Then, we run the peeling process similar as \newa. Note that, for an unassigned edge $e \in \ggei$, we assign $\bts_e$ to $e$ and mark $e$ as assigned only if $\btf_e \geq \ei$. This is because in each iteration, \newap only computes the bitruss numbers for the edges in the $\epsilon$-bitruss.

\noindent
{\textbf{Step 3: preparation for the next iteration.}}
For an iteration $i$, after running steps 1 - 2, we can obtain the bitruss number of all the edges in $H_{\ei}$. Then, we decrease $\ei$ and run steps 1 - 2 until all the edges are assigned. To reduce the number of iterations, we can decrease $\ei$ by an integer larger than 1, which means that we compute $\eii = \ei - \alpha$ where $\alpha \geq 1$ is an integer. Consequently, we can take all the edges with butterfly supports $\geq \eii$ into consideration and compute the $k$-bitrusses with $\eii \leq k < \ei$ in one iteration. In \newap, we set $\alpha$ as $\lceil k_{max} \times \tau \rceil$ where $k_{max}$ is the largest possible bitruss number and $\tau \in (0, 1]$. Thus, the total number of iterations in \newap can be reduced from $k_{max}$ to $ \lceil \frac{k_{max}}{\lceil k_{max} \times \tau \rceil} \rceil$. We also provide a guideline of choosing $\tau$ in Section \ref{sct:experiment}.

\begin{algorithm}[th]
\small
\DontPrintSemicolon
\KwIn{$\ggei$: the candidate graph in iteration $i$}
\KwOut{$\ggei$: the \bei of $\ggei$} 
//$\btf_e$ for each $e \in E(\ggei)$ is pre-computed\;
Compute $p(u)$ for each $u \in V(\ggei)$  // Definition \ref{def:priority} \;
\ForEach{ $u \in V(\ggei)$ } {
    run Algorithm \ref{algo:ic} lines 4 - 7, replace $G$ with $\ggei$\;
    \ForEach{$v \in \nbp(u): p(v) < p(u)$} {
        \ForEach{$w \in \nbp(v): p(w) < p(u)$} {
            \If {$count\_wedge(w) > 1$} {
                run Algorithm \ref{algo:ic} 11-14, replace $I$ with $\igei$\;
                \If {unassigned $e = (u, v)\ \notin L(\igei)$} {
                    add\ $e.id$ and $\btf_e$ into $L(\igei)$\;
                }
                \If {unassigned $e = (v, w)\ \notin L(\igei)$} {
                    add\ $e.id$ and $\btf_e$ into\ $L(\igei)$\;
                }
                \If {$(u, v)$ or $(v, w)$ \ is\ unassigned} {
                    run Algorithm \ref{algo:ic} lines 19 - 22, replace $I$ with $\igei$\;
                }
            }
        }
    }
}
\Return{$\igei$}\; 
\caption{{CompressedIndexConstruction}}
\label{algo:ica}
\end{algorithm}

\begin{algorithm}[hbt]
\small
\DontPrintSemicolon
\KwIn{$G(V = (U, L), E)$: the input bipartite graph, $\tau \in (0, 1]$}
\KwOut{$\bts_e$ for each $e \in E(G)$}
compute $\btf_e$ for each $e \in E(G)$\;
compute the largest possible bitruss number $k_{max}$ in $G$\;

$i \gets 1$; $\epsilon_i \gets k_{max}$\;
\While{exist unassigned edge in $G$} {
    extract $\ggei$ from $G$ where $\btf_e \geq \ei$ for each $e \in E(\ggei)$\;
    recompute $\btf_e$ for each $e \in E(\ggei)$ on $\ggei$ and remove $e$ from $\ggei$ if $\btf_e < \ei$\;
    call CompressedIndexConstruction // Algorithm \ref{algo:ica}\;
    \While{exist unassigned edges in $\ggei$} {
        run Algorithm 5 lines 3 - 21, replace $I$, $G$ with $\igei$, $\ggei$\;
    }
    $\alpha \gets \lceil k_{max} \times \tau \rceil$\;
    $i \gets i + 1$; $\eii \gets \max\{\ei - \alpha, 0\}$\;
}

\Return{$\bts_e$ for each $e \in E(G)$}
\caption{{\sc \newap}}
\label{algo:newap}
\end{algorithm}


\noindent
{\textbf{Analysis of the \newap algorithm.}}
Below we show the correctness and time/space complexities of \newap.

\begin{theorem}
\label{theorem:newc}
The \newap algorithm correctly solves the \btsd problem.
\end{theorem}

\begin{proof}
\label{proof:newc}
This theorem immediately follows from Lemma \ref{lemma:kbitruss} and Theorem \ref{theorem:newco}. 
\end{proof}

\noindent
{\bf Time complexity.}
The time complexity of the \newap algorithm is $O(\sum_{(u, v) \in E(G)}\min\{\degree(u), \degree(v)\} + \sum_{ i \in [1, t]}({\sum_{(u, v) \in E(\ggei)}\min\{\degree(u), \degree(v)\} + \btf_{\ggei} - \btf_{F_i}}))$ where $\btf_{F_i}$ denotes the number of butterflies containing the set $F_i$ of assigned edges in $\ggei$, and $t =  \lceil \frac{k_{max}}{\lceil k_{max} \times \tau \rceil} \rceil$ is the total number of iterations. The details of analysis are as follows. In \newap, the counting process needs $O(\sum_{(u, v) \in E(G)}\min\{\degree(u), \degree(v)\}$ time \cite{wang2019vertex}. Also, for each iteration $i$, we need $O(\sum_{(u, v) \in E(\ggei)}\min\{\degree(u), \degree(v)\}$ time to compute the \bei of the candidate subgraph $\ggei$. The peeling process in iteration $i$ needs $O(\btf_{\ggei} - \btf_{F_i})$ time since we can avoid updating the edges which were already assigned in previous iterations.



\noindent
{\bf Space complexity.} For \newap, we need $O(m)$ space to store the bitruss numbers for edges. Also, in iteration $i$, when handling a candidate subgraph $\ggei$, we need to construct a compressed \bei for $\ggei$ and release it before the next iteration. Since $\ggei \subseteq G$ and the \bei of $\ggei$ uses $O(\sum_{(u, v) \in E(\ggei)}\min\{\degree(u), \degree(v)\})$ space, the space complexity of \newap is $O(\sum_{(u, v) \in E(G)}\min\{\degree(u), \degree(v)\})$.


%
%

\begin{table*}[ht]
\small
\caption{\label{table:datasets}
Summary of Datasets}
\vspace*{-4mm}
\begin{center}
\scalebox{0.88}{
\begin{tabular}{c|c|c|c|c|c|c}
\noalign{\hrule height 1pt}
\textbf{Dataset}& $|E|$  & $|U|$ & $|L|$  & $\btf_G$ & $\btf_{e_{max}}$  & $\bts_{e_{max}}$\\
\hline
\hline
Condmat & 58{,}595 & 16{,}726 & 22{,}015 &  70{,}549 &  127 & 63\\
\hline
Marvel & 96{,}662 & 6{,}486 & 12{,}942 & 10{,}709{,}594 &  6{,}612 & 1761\\
\hline
DBPedia & 293{,}697 & 172{,}091 & 53{,}407& 3{,}761{,}594 & 1{,}720 &852\\
\hline
Github & 440{,}237 & 56{,}519 & 120{,}867 & 50{,}894{,}505 & 40{,}675 & 1014\\
\hline
Twitter& 1{,}890{,}661 & 175{,}214 & 530{,}418 & 206{,}508{,}691  & 29{,}708 &5864\\
\hline
D-label & 5{,}302{,}276& 1{,}754{,}823 & 270{,}771 & 3{,}261{,}758{,}502 & 625{,}418 & 15498\\
\hline
D-style & 5{,}740{,}842 & 1{,}617{,}943 & 383 & 77{,}383{,}418{,}076 &  1{,}279{,}105 & 52{,}015\\
\hline
Amazon& 5{,}743{,}258 & 2{,}146{,}057 & 1{,}230{,}915 & 35{,}849{,}304 & 8{,}827 & 551\\
\hline
DBLP & 8{,}649{,}016 & 4{,}000{,}150 & 1{,}425{,}813 & 21{,}040{,}464 & 641 & 420\\
\hline
Wiki-it & 12{,}644{,}802 & 2{,}225{,}180 & 137{,}693 & 298{,}492{,}670{,}057 &  2{,}994{,}802& 166{,}785\\
\hline
Wiki-fr& 22{,}090{,}703  & 288{,}275 & 4{,}022{,}276   & 601{,}291{,}038{,}864 &4{,}500{,}590&231{,}253\\
\hline
Delicious  & 101{,}798{,}957 & 833{,}081 & 33{,}778{,}221   & 56{,}892{,}252{,}403 & 1{,}219{,}319&6{,}638\\
\hline
Live-journal & $112{,}307{,}385$ & 3{,}201{,}203 & 7{,}489{,}073   & 3{,}297{,}158{,}439{,}527 & 10{,}025{,}933& 456{,}791\\
\hline
Wiki-en  & 122{,}075{,}170 & 3{,}819{,}691 & 21{,}504{,}191   & 2{,}036{,}443{,}879{,}822 & 18{,}206{,}363&438{,}728\\
\hline
Tracker  & 140{,}613{,}762 & 27{,}665{,}730 & 12{,}756{,}244   & 20{,}067{,}567{,}209{,}850 & 46{,}747{,}317&2{,}462{,}013\\
\noalign{\hrule height 1pt}

\end{tabular}}
\end{center}
\vspace{-8mm}
\end{table*}

\section{Experiments}
\label{sct:experiment}
In this section, we report the evaluation of bitruss decomposition algorithms on 15 real-world datasets.
\subsection{Experiments setting}
\noindent
{\bf Algorithms.}
Our empirical studies have been conducted against the following algorithms:
1) the state-of-the-art \bs in \cite{sariyuce2018peeling} deployed with the new counting algorithm in \cite{wang2019vertex} as the baseline algorithm,
2) the bottom-up algorithm \new in Section \ref{sct:algorithms}, 3) the bottom-up algorithm with batch-based optimizations \newa in Section \ref{sct:algorithms}, 4) the progressive compression algorithm \newap in Section \ref{sct:algorithms}.

The algorithms are implemented in C++ and the experiments are run on a Linux server with Intel Xeon E5-2698 processor and 512GB main memory.
{\em We terminate an algorithm if the running time is more than 30 hours}.

\noindent
{\bf Datasets.}
We use $15$ real datasets in our experiments and all the datasets we used can be found in KONECT \footnote{http://konect.uni-koblenz.de}.

The summary of datasets is shown in Table \ref{table:datasets}. $U$ and $L$ are vertex layers, $|E|$ is the number of edges. $\btf_G$ is the number of butterflies. $\btf_{e_{max}}$ and $\bts_{e_{max}}$ are the largest butterfly support and largest bitruss number of an edge in a dataset, respectively.

\noindent
{\bf Parameters.} The experiments are conducted using different settings on 2 parameters: $n$ (graph size), $\tau$ (the parameter used in \newap). When varying the graph size $n$, we randomly sample 20\% to 100\% vertices of the original graphs, and construct the induced subgraphs using these vertices. We vary $\tau$ from 0.02 to 1 and set $\tau$ as 0.02 by default.

\subsection{Performance Evaluation}
In this section, we evaluate the performance of the proposed algorithms.
First, we evaluate the performance of \bs, \new, \newa and \newap on all the datasets. After that, we evaluate the number of butterfly support updates and the size of online indexes of \new, \newa and \newap. Then, we test the scalability of our algorithms. Also, we evaluate the batch-based optimizations. Finally, we evaluate the parameter $\tau$ used in \newap.
%
%
\begin{figure*}[htb]
\begin{centering}
\includegraphics[trim=0 0 0 0,clip,width=0.85\textwidth]{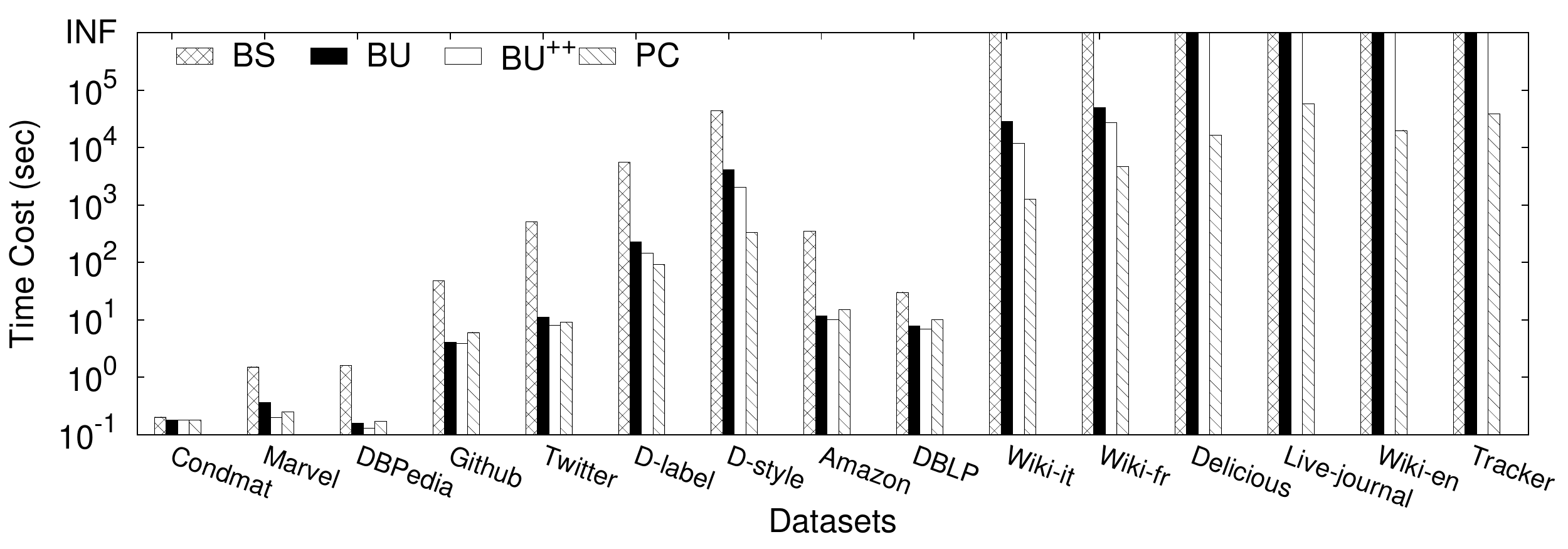}
\vspace*{-3mm}\caption{Performance on different datasets}
\label{fig:all}
\vspace*{-7mm}
\end{centering}
\end{figure*}

\noindent
{\bf Evaluating the performance on all the datasets.}
In Figure \ref{fig:all}, we show the performance of the \bs, \new, \newa and \newap algorithms on different datasets. We can observe that \new, \newa and \newa outperform \bs on all the datasets. This is because these algorithms utilize the \bei which significantly reduces the computation cost of edge removal operations. The \new, \newa and \newap algorithms are all at least one order of magnitude faster than the \bs algorithm on \texttt{DBPedia}, \texttt{Twitter}, \texttt{D-label}, \texttt{D-style} and \texttt{Amazon}. Especially, on \texttt{D-style}, the \newap algorithm is at least two orders of magnitude faster than the \bs algorithm. On \texttt{Wiki-it} and \texttt{Wiki-fr}, the \bs algorithm cannot finish within 30 hours while all our algorithms can. We can also observe that \newap is slightly slower than \new and \newa on some datasets such as \texttt{Amazon} and \texttt{DBLP}. This is because $\btf_{e_{max}}$ and $\bts_{e_{max}}$ are small on these datasets, the time cost of butterfly support updates that can be reduced by \newap is limited and \newap needs additional pre-processing time in each iteration. However, only \newap can finish on large-scale datasets \texttt{Delicious}, \texttt{Live-journal}, \texttt{Wiki-en} and \texttt{Tracker} within 30 hours. This is because \newap handles and compresses large-scale graphs progressively and significantly reduces the number of butterfly support updates for edges especially for the hub edges.  In the following experiments, we omit the \bs algorithm since all our algorithms significantly outperform \bs as evaluated here.
%
%
%


\begin{figure}[htb]
\begin{centering}
\includegraphics[trim=0 0 0 0,clip,width=0.48\textwidth]{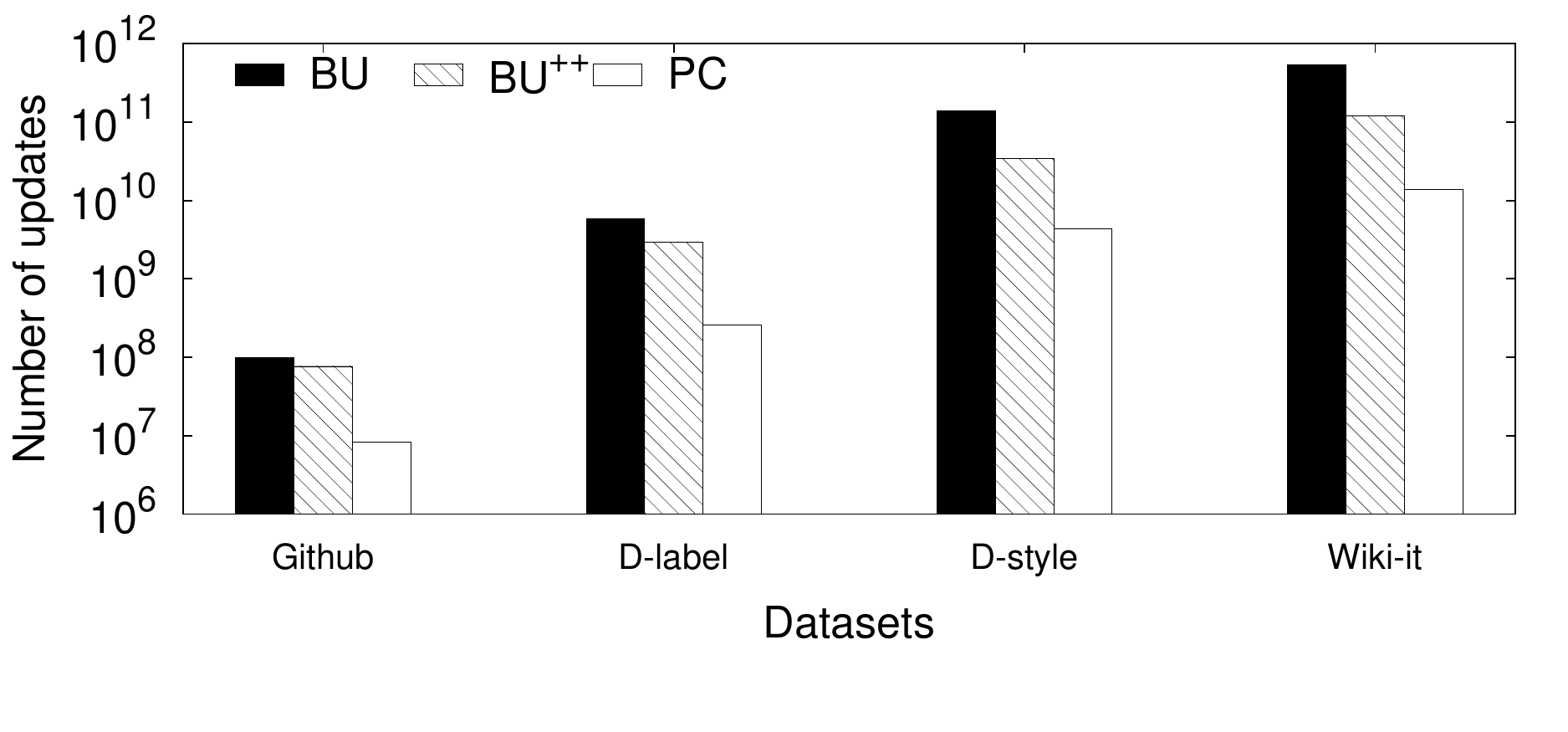}
\vspace*{-3mm}
\caption{The total number of butterfly support updates}
\label{fig:update_all}
\end{centering}
\end{figure}


\noindent
{\bf Evaluating the total number of butterfly support updates.}
In Figure \ref{fig:update_all}, we show the number of updates of our algorithms on four representative datasets \texttt{Github}, \texttt{D-label}, \texttt{D-style} and \texttt{Wiki-it}. Here, the number of updates means the total number of butterfly support updates for edges (i.e., the sum of updates of $\btf_e$ for each edge $e$). We can observe that on all these datasets, the number of updates of \newa is less than the number of updates of \new because of the batch-based optimizations. \newap reduces more than 90\% updates than \new and \newa. This is because \newap processes the hub edges within a more cohesive subgraph and generates compressed \bei in later iterations. Thus, it significantly reduces the number of butterfly support updates for those hub edges as discussed in Section \ref{sct:newap}.

\begin{figure}[htb]
\begin{centering}
\includegraphics[trim=0 0 0 0,clip,width=0.48\textwidth]{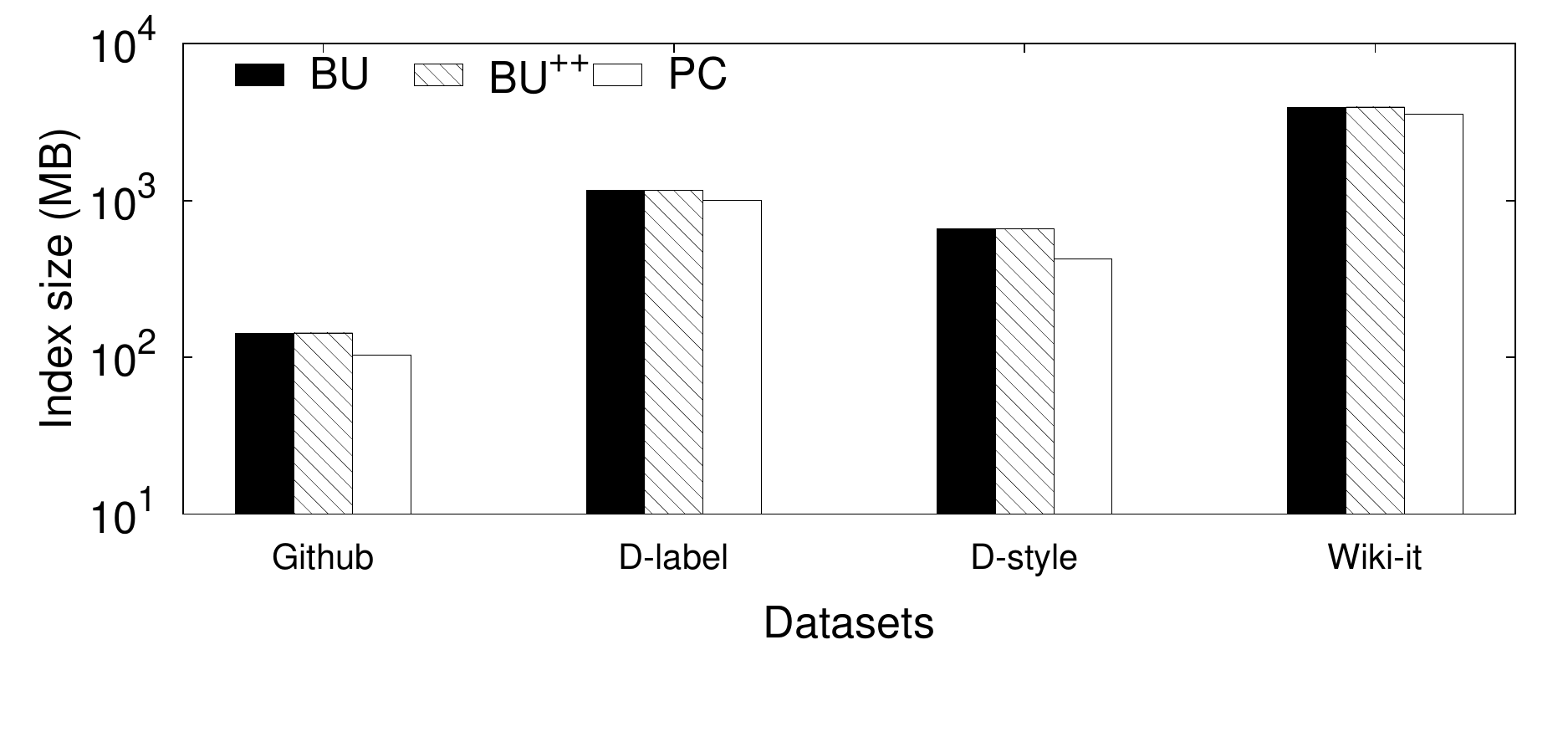}
\vspace*{-3mm}
\caption{The size of online indexes}
\label{fig:update_all}
\end{centering}
\end{figure}

\noindent
{\bf Evaluating the size of online indexes.}
In Figure \ref{fig:index}, we show the size of online indexes constructed by our algorithms on four representative datasets \texttt{Github}, \texttt{D-label}, \texttt{D-style} and \texttt{Wiki-it}. We can observe that on each dataset, the size of online index of \newap is less than the size of online indexes of \new and \newa. This is because \newap processes a more cohesive subgraph and generates compressed \bei in each iterations. In addition, we can see that, on \texttt{Wiki-it} with $601,291$ million butterflies, the online indexes of all our algorithms only need less than 4GB space.

\begin{figure}[htb]
\begin{centering}
\vspace*{-2mm}\subfigure[\texttt{Github}, varying $n$]{
\begin{minipage}[b]{0.22\textwidth}
\includegraphics[trim=0 0 0 0,clip,width=1\textwidth]{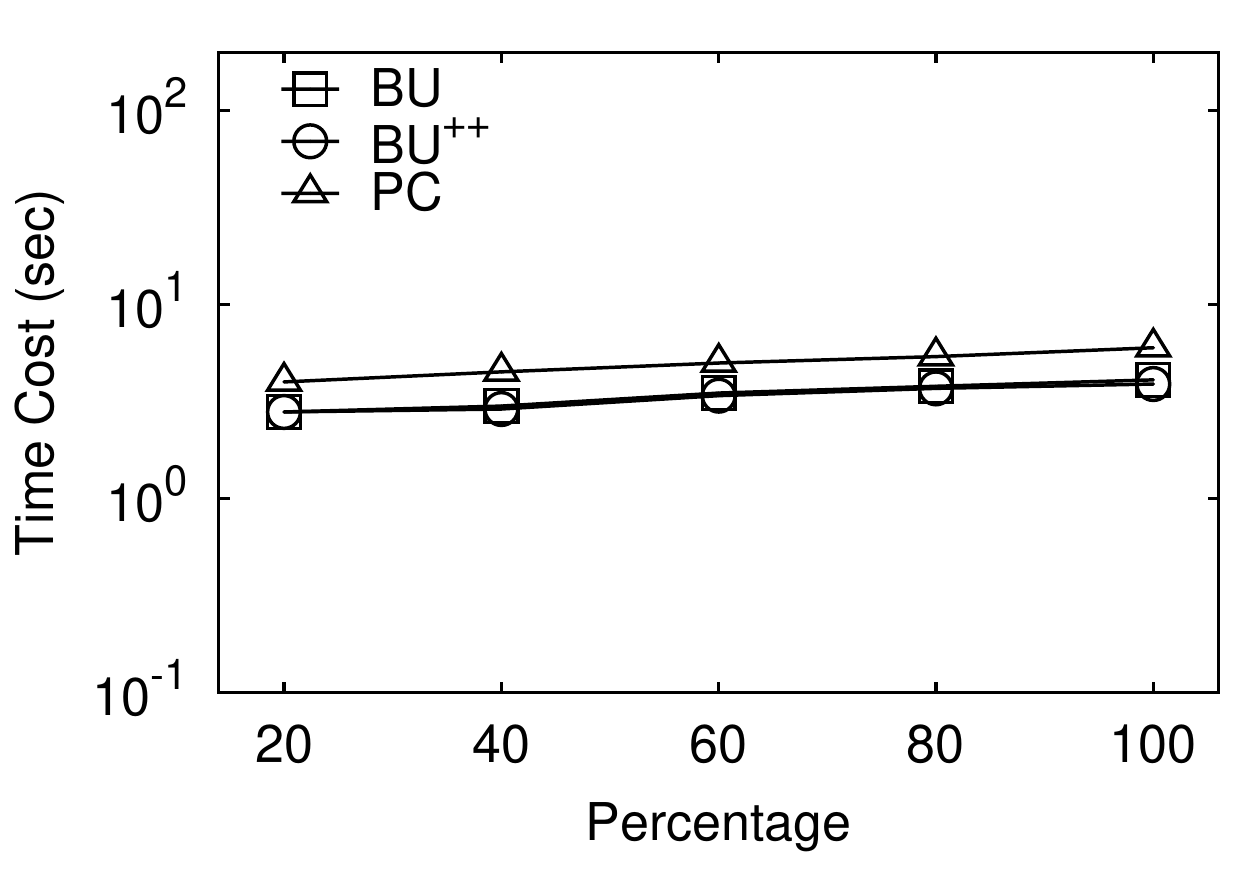}\vspace{-5.5mm}
\label{fig:n1}
\end{minipage}}
\vspace*{-2mm}\subfigure[\texttt{D-label}, varying $n$]{
\begin{minipage}[b]{0.22\textwidth}
\includegraphics[trim=0 0 0 0,clip,width=1\textwidth]{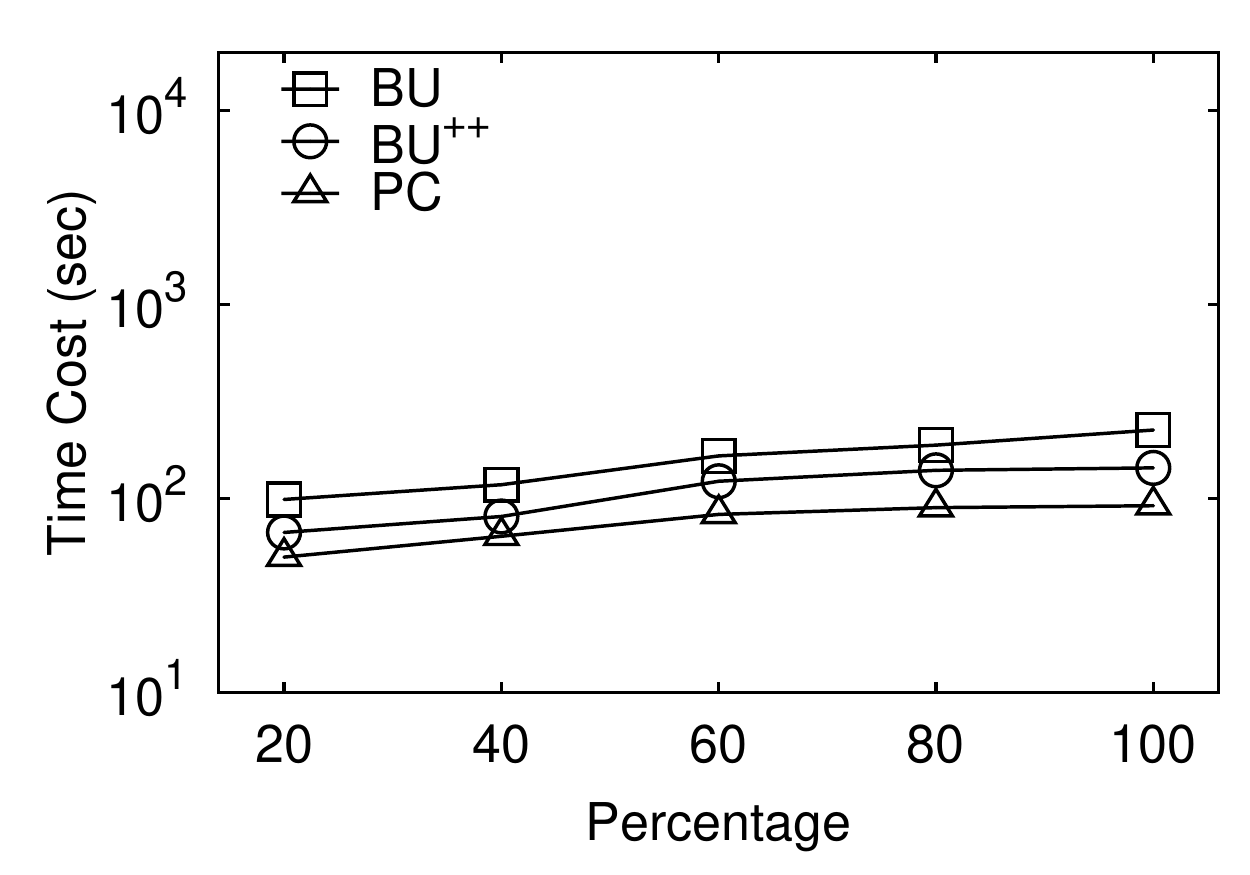}\vspace{-5.5mm}
\label{fig:n2}
\end{minipage}}
\subfigure[\texttt{D-style}, varying $n$]{
\begin{minipage}[b]{0.22\textwidth}
\includegraphics[trim=0 0 0 0,clip,width=1\textwidth]{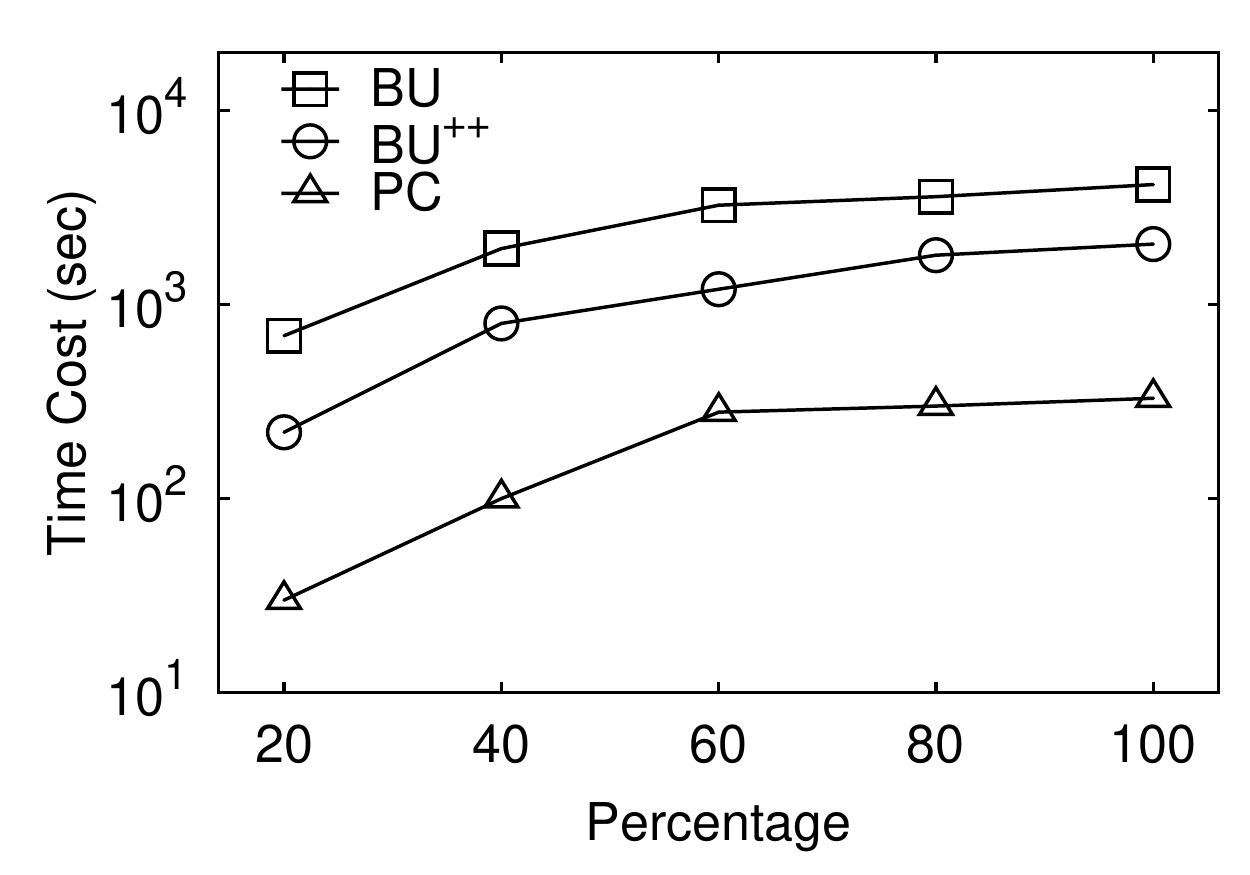}\vspace{-5.5mm}
\label{fig:n1}
\end{minipage}}
\subfigure[\texttt{Wiki-it}, varying $n$]{
\begin{minipage}[b]{0.22\textwidth}
\includegraphics[trim=0 0 0 0,clip,width=1\textwidth]{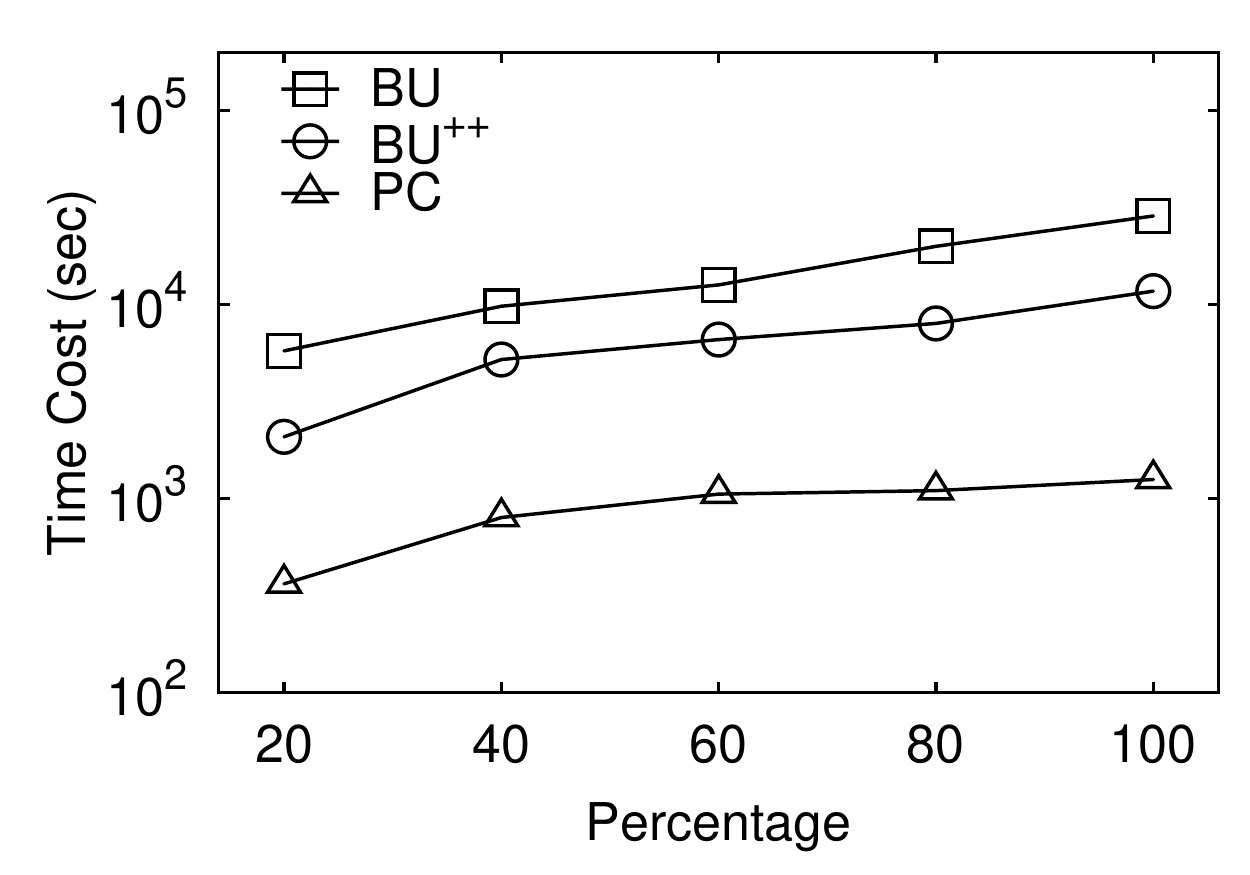}\vspace{-5.5mm}
\label{fig:n2}
\end{minipage}}
\vspace*{-1mm}\caption{Effect of graph size}
\label{fig:n}
\end{centering}
\end{figure}

\noindent
{\bf Scalability.} {\em Evaluating the effect of graph size.} In Figure \ref{fig:n}, we study the scalability of \new, \newa and \newap by varying the graph size $n$ on the \texttt{Github}, \texttt{D-label}, \texttt{D-style} and \texttt{Wiki-it} datasets. When varying $n$, we randomly sample 20\% to 100\% vertices of the original graphs, and construct the induced subgraphs using these vertices. We can observe that, the algorithms \new, \newa and \newap are scalable. The computation costs of all these algorithms increase as the percentage of vertices increases. As discussed before, \newap significantly outperforms \new and \newap on \texttt{D-style} and \texttt{Wiki-it}.

\begin{figure}[htb]
\begin{centering}
\includegraphics[trim=0 0 0 0,clip,width=0.48\textwidth]{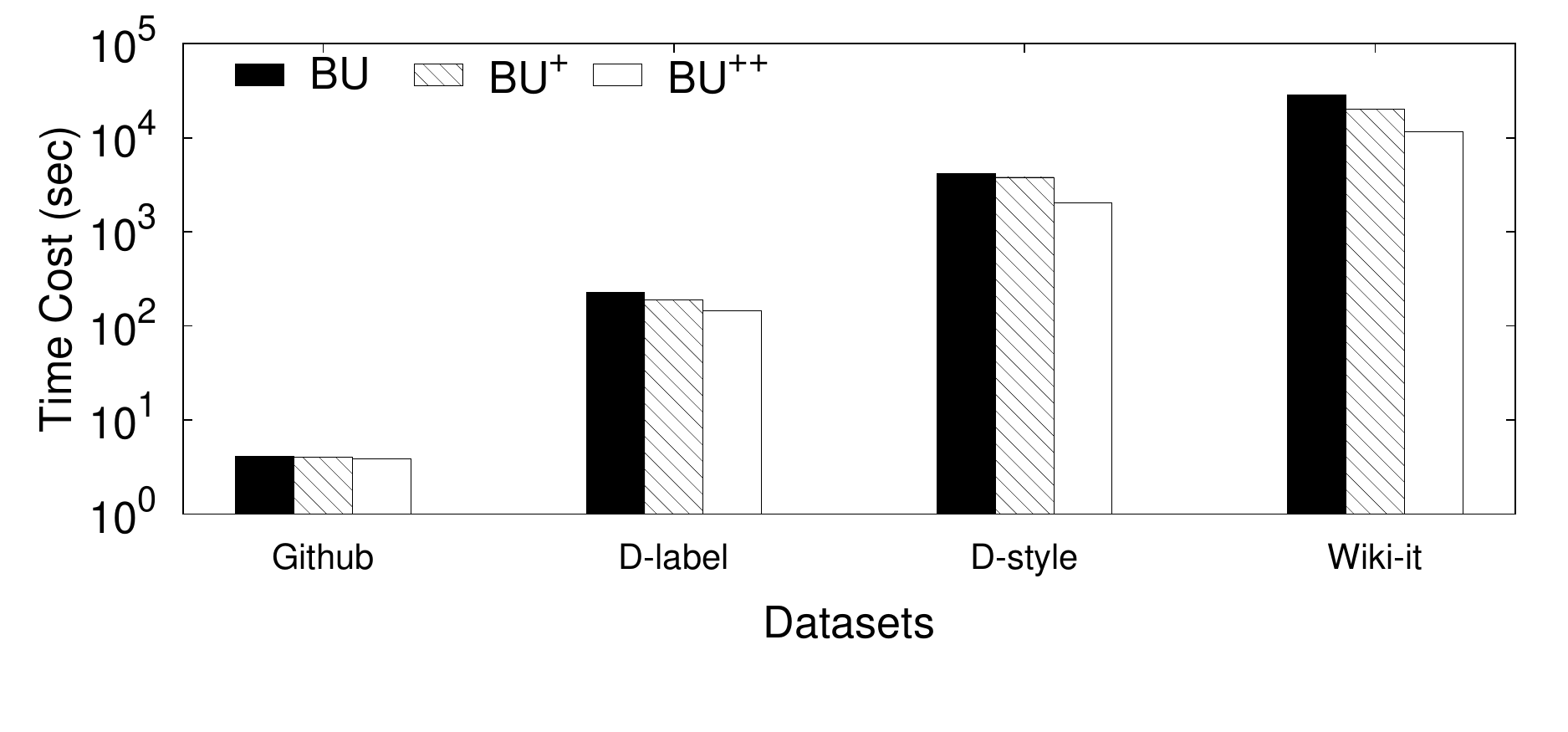}
\vspace*{-3mm}
\caption{Effect of the batch-based optimizations}
\label{fig:batch}
\end{centering}
\end{figure}

\noindent
{\bf Evaluate the batch-based optimizations.}
In Figure \ref{fig:batch}, we evaluate the efficiency of our batch-based optimizations (i.e., batch edge processing and batch bloom processing in Section \ref{sct:newa}) on \texttt{Github}, \texttt{D-label}, \texttt{D-style} and \texttt{Wiki-it} datasets. We obtain \newm by combining the batch edge processing optimization with \new, and \newa is \new combining with both these two batch-based optimizations. We can see that, the batch edge processing optimization significantly reduces the computation cost while the batch bloom processing optimization further enhances the performance.

\begin{figure}[htb]
\begin{centering}
\vspace*{-2mm}\subfigure[Time cost]{
\begin{minipage}[b]{0.22\textwidth}
\includegraphics[trim=0 0 0 0,clip,width=1\textwidth]{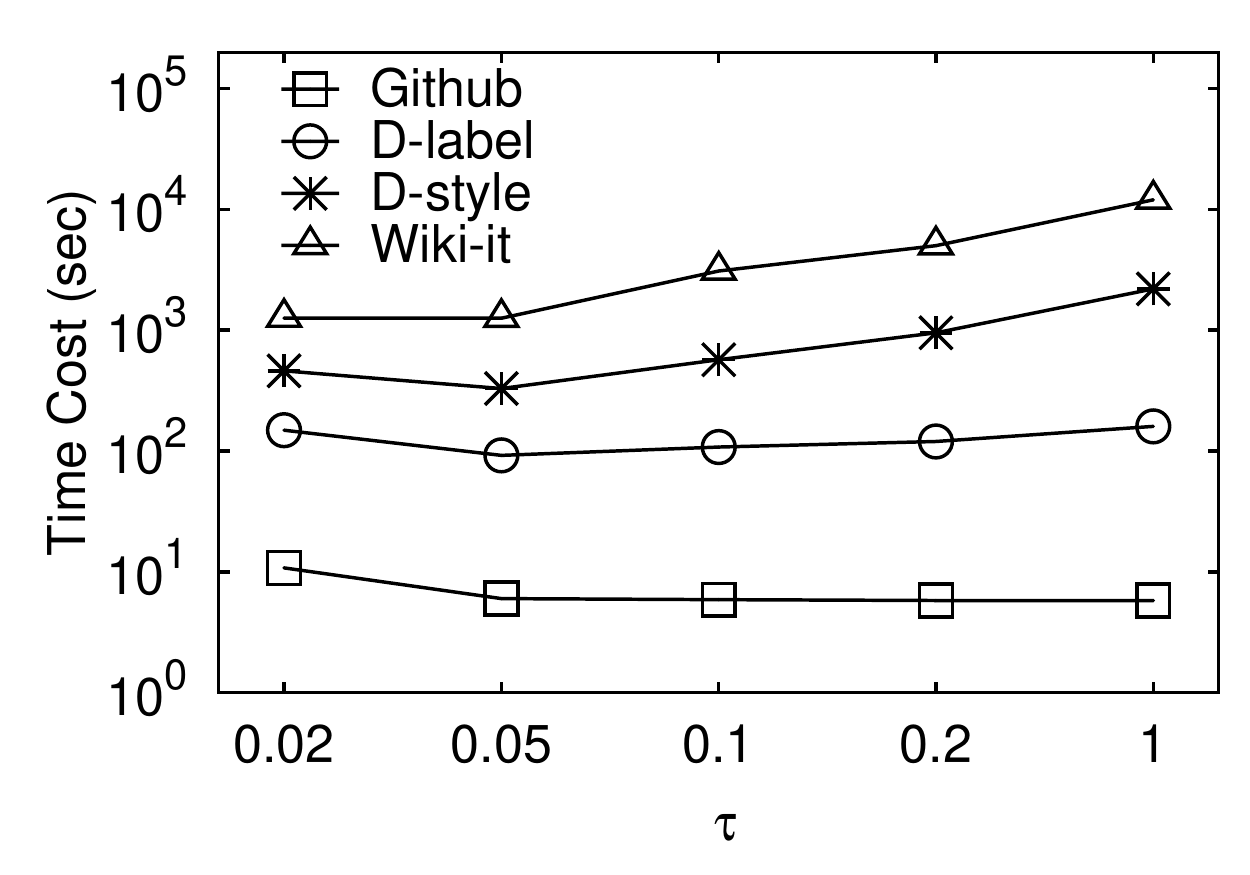}\vspace{-5.5mm}
\label{fig:p1}
\end{minipage}}
\vspace*{-2mm}\subfigure[Number of updates]{
\begin{minipage}[b]{0.22\textwidth}
\includegraphics[trim=0 0 0 0,clip,width=1\textwidth]{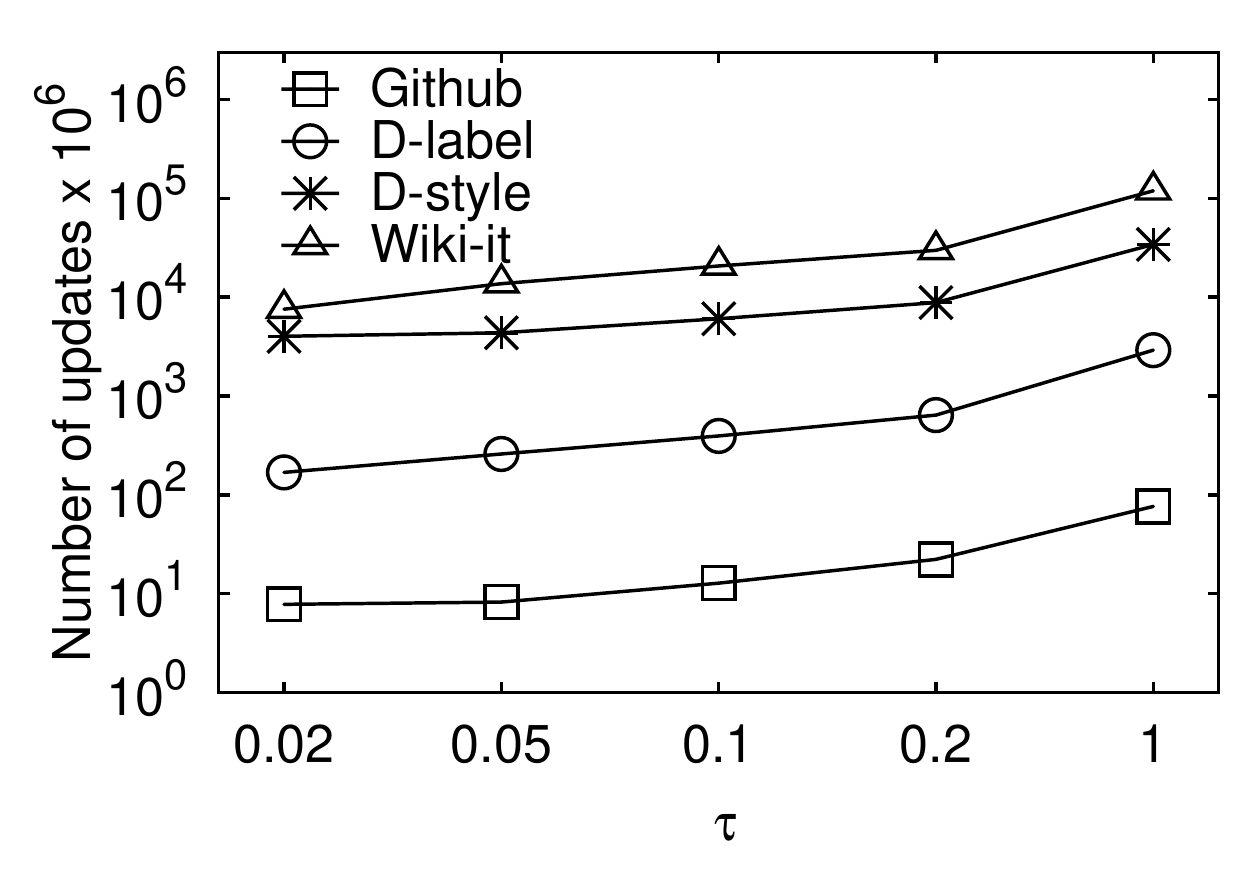}\vspace{-5.5mm}
\label{fig:p2}
\end{minipage}}
\caption{Effect of $\tau$}
\label{fig:tau}
\vspace*{-1mm}
\end{centering}
\end{figure}

\noindent
{\bf Evaluate the effect of $\tau$.}
The algorithm \newap needs a parameter $\tau$ to decide the decrease of $k'$ after each iteration of processing. In Figure \ref{fig:tau}, we evaluate the effect of $\tau$ on datasets \texttt{Github}, \texttt{D-label}, \texttt{D-style} and \texttt{Wiki-it}. Figure \ref{fig:tau} (b) shows the number of updates increases when $\tau$ increases. This is because when $\tau$ is small, the original graph will be compressed much more times and the number of updates for hub edges decreases. Although using a smaller $\tau$ may lead to fewer updates, it will increase the number of iterations and additional computation is needed for each iteration as discussed in Section \ref{sct:newap}. In Figure \ref{fig:tau} (a), we can see that, the efficiency of $\newap$ is not very sensitive to $\tau$ when $\tau$ is larger than 0.02 on \texttt{Github} and \texttt{D-label}. For large datasets such as \texttt{D-style} and \texttt{Wiki-it}, there is a local minimum in Figure \ref{fig:tau} (a). According to the experimental result, we suggest to set $\tau$ to $0.05 - 0.2$.

\vspace{-0.1cm}
\section{Conclusion}

\label{sct:conclusion}
In this paper, we study the \btsd problem. To solve this problem efficiently, we propose a novel online \bei which compresses the butterflies into blooms. Based on the \bei, we first propose a bottom-up algorithm \new which reduces the time complexities of the existing algorithms. Also, two batch-based optimizations are deployed on \new to enhance the performance. Then, to efficiently handle edges with high butterfly supports, we propose the \newap algorithm which handles and compresses the graph progressively. We conduct extensive experiments on real datasets and the result shows that our algorithms significantly outperform the state-of-the-art algorithm.
\vspace{-0.1cm}



{\small
\bibliographystyle{IEEEtran}
\bibliography{paper}
}

\end{document}